\documentclass[12pt]{article}

\usepackage{bbm}
\usepackage{amssymb,amsmath,amsthm}
\usepackage[dvips,xdvi]{graphics}
\usepackage[dvips,xdvi]{graphicx}
\usepackage{textcomp}
\usepackage{array,longtable}

\allowdisplaybreaks

\newcommand{\lang}{\langle\langle}
\newcommand{\rang}{\rangle\rangle}

\newcommand{\zed}{\mathcal{Z}}
\newcommand{\one}{\mathbbm{1}}

\swapnumbers 
\theoremstyle{definition} 
\newtheorem{define}{Definition}[section]

\theoremstyle{definition} 
\newtheorem{prop}[define]{Proposition} 

\theoremstyle{definition}
\newtheorem{lemma}[define]{Lemma}

\theoremstyle{definition}
\newtheorem{cor}[define]{Corollary}

\theoremstyle{definition}
\newtheorem{crit}[define]{Criterion}

\theoremstyle{definition}
\newtheorem{theorem}[define]{Theorem}

\begin{document}

\title{
\normalsize \hfill UWThPh-2010-11 \\*[16mm]
\LARGE On the finite subgroups of U(3) \\
of order smaller than 512
}

\author{
Patrick Otto Ludl\thanks{E-mail: patrick.ludl@univie.ac.at}
\\*[3mm]
\small University of Vienna, Faculty of Physics \\
\small Boltzmanngasse 5, A--1090 Vienna, Austria
}

\date{20 December 2010}

\maketitle

\begin{abstract}
We use the SmallGroups Library to find the finite subgroups of $U(3)$ of order smaller than 512 which possess a faithful three-dimensional irreducible representation. From the resulting list of groups we extract those groups that can not be written as direct products with cyclic groups. These groups are important building blocks for models based on finite subgroups of $U(3)$.
\\
All resulting finite subgroups of $SU(3)$ can be identified using the well known list of finite subgroups of $SU(3)$ derived by Miller, Blichfeldt and Dickson at the beginning of the 20\textsuperscript{th} century. Furthermore we prove a theorem which allows to construct infinite series of finite subgroups of $U(3)$ from a special type of finite subgroups of $U(3)$. This theorem is used to construct some new series of finite subgroups of $U(3)$. The first members of these series can be found in the derived list of finite subgroups of $U(3)$ of order smaller than 512.
\\
In the last part of this work we analyse some interesting finite subgroups of $U(3)$, especially the group $S_4(2)\cong A_4\rtimes\zed_4$, which is closely related to the important $SU(3)$-subgroup $S_4$.
\end{abstract}

\newpage

\section{Introduction}
The problem of lepton masses and mixing (more generally the fermion mass and mixing problem) is one of the most interesting current research topics of particle physics and withstood a solution for decades. Invariance of the Lagrangian under finite family symmetry groups constitutes an interesting possibility, at least for a partial solution of this problem.
\medskip
\\
In 2002 Harrison, Perkins and Scott suggested the tribimaximal lepton
mixing matrix \cite{HPS}
	\begin{equation}
	U_{\mathrm{TBM}}=\left(\begin{matrix}
		 \sqrt{\frac{2}{3}} & \frac{1}{\sqrt{3}} & 0 \\
		 -\frac{1}{\sqrt{6}} & \frac{1}{\sqrt{3}} &
                 \frac{1}{\sqrt{2}} \\ 
		 \frac{1}{\sqrt{6}} & -\frac{1}{\sqrt{3}} & \frac{1}{\sqrt{2}}
	\end{matrix}\right),
	\end{equation}
which is in agreement with the current
experimental bounds on the lepton mixing matrix \cite{results}. The nice
appearance of this matrix induced the idea of an underlying symmetry
in the Lagrangian of the lepton and scalar sector. A large number of
models involving especially discrete symmetries followed. For a review
of today's state of the art in model building see \cite{models}.
\medskip
\\
Due to the fact that there are three known families of leptons (and
quarks) models based on finite subgroups of $U(3)$ have become very
popular, so a systematic analysis of the finite subgroups of $U(3)$ would provide an invaluably helpful tool for model building with finite family symmetry groups.
\\
Unfortunately the finite subgroups of $U(3)$ have, to our knowledge, not been classified by now, while the finite subgroups of $SU(3)$ have been classified already at the beginning of the 20\textsuperscript{th} century by Miller, Blichfeldt and Dickson \cite{miller}.
\medskip
\\
The idea of a systematic analysis of finite subgroups of
$SU(3)$ in the context of particle physics is not new
\cite{finitesubgroups-su3-fairbairn,BLW1,BLW2,Fairbairn2}, and research on the application of finite
groups in particle physics (especially finite subgroups of $SU(3)$)
continues unabated \cite{Frampton,D3n^2,SimpleFinite,D6n^2,Dihedral,DTPOL,Zwicky,Ishimori:2010au}.

\section{The finite subgroups of U(3) of order smaller than 512}

In this work we want to concentrate on the finite subgroups of $U(3)$. We will use both the analytical tools of group theory as well as the modern tool of computer algebra to investigate the finite subgroups of $U(3)$ up to order 511.

\subsection{Classification of finite subgroups of $U(3)$}

At first let us consider the different types of finite subgroups of $U(3)$ we will distinguish. Knowing that every representation of a finite group is equivalent to a unitary representation we find:
\begin{quote}
A finite group $G$ is isomorphic to a finite subgroup of $U(3)$ if and only if it possesses a faithful three-dimensional representation.
\end{quote}
Thus, if we would search for all finite groups which fulfil the above properties, we would obtain all finite subgroups of $U(3)$, especially we would obtain all finite subgroups of $U(1)$ and $U(2)$ too. When we usually speak of ``finite subgroups of $U(3)$'' we primarily mean those finite subgroups of $U(3)$ which are not finite subgroups of $U(1)$ and $U(2)$. 
\\
At this point it is important to notice that the possession of a faithful 3-dimensional \textit{irreducible} representation is sufficient but not necessary for a group to be a finite subgroup of $U(3)$\footnote{The author wants to thank K.M.~Parattu and A.~Wingerter for pointing this out in their paper \cite{parattu-wingerter}.}. In fact there are many finite subgroups of $U(3)$ which possess a faithful 3-dimensional \textit{reducible} representation but do not possess any faithful 1- or 2-dimensional representations. According to \cite{parattu-wingerter} these groups correspond to finite subgroups of $U(2)\times U(1)$. 
\\
Though doing so we will miss the $U(2)\times U(1)$-subgroups mentioned above, we will in this work specialise to the finite subgroups of $U(3)$ which possess a faithful 3-dimensional irreducible representation. Among these groups we can differentiate
	\begin{itemize}
	 \item groups which have a faithful irreducible representation of determinant\footnote{Let $D:G\rightarrow D(G)$ be a representation of a finite group $G$. We say $D$ has determinant 1, if all matrices in $D(G)$ have determinant 1.} 1 ($SU(3)$-subgroups) and
	 \item groups which don't have a faithful irreducible representation of determinant 1.
	\end{itemize}
An analysis similar to the one performed in this work can be found in \cite{parattu-wingerter}, where all finite groups up to order 100 are listed. The list given in \cite{parattu-wingerter} especially indicates which groups of order up to 100 possess 3-dimensional faithful representations (both reducible and irreducible). 

\subsection{The computer algebra system GAP and the SmallGroups Library}
The task of classifying all finite groups which have a faithful three-dimensional irreducible representation would need an analysis using the techniques described in \cite{miller} applied to $U(3)$. Instead of doing that, we want to get a first impression of the finite subgroups of $U(3)$ by searching for $U(3)$-subgroups of small orders with the help of the computer algebra system GAP \cite{GAP}. Through the SmallGroups package \cite{SmallGroups} GAP allows access to the SmallGroups Library \cite{SmallGroups,SmallGroupsLibrary} which contains, among other finite groups, all finite groups up to order 2000, except for the groups of order 1024, up to isomorphism.
\medskip
\\
The way finite groups are labeled in the SmallGroups Library is the following: Let there be $n$ non-isomorphic groups of order $g$, then these $n$ groups are labeled by their order $g$ and a number $j\in\{1,...,n\}$:
	\begin{quote}
	$\mbox{\textlbrackdbl}g,j\mbox{\textrbrackdbl}$ denotes the $j-$th finite group of order $g$ $(j\in\{1,...,n\})$ listed in the SmallGroups Library.
	\end{quote}
The way in which the $n$ non-isomorphic groups of a given order $g$ are arranged in the SmallGroups Library depends on $g$. For a detailed description of the SmallGroups Library we refer the reader to chapter 48.7 of the GAP reference manual \cite{GAP}. For our purpose we only need to know that $\mbox{\textlbrackdbl}g,j\mbox{\textrbrackdbl}$ is not isomorphic to $\mbox{\textlbrackdbl}g,k\mbox{\textrbrackdbl}$ if $j\neq k$. To all ``small groups'' $\mbox{\textlbrackdbl}g,j\mbox{\textrbrackdbl}$ listed in this paper the reader can find the common name or a list of generators in tables \ref{SU3-subgroups} and \ref{U3-subgroups}, respectively.
\\
Let us now get a picture of the number of finite groups of some given order. Figure \ref{numbergroups} shows the total number $N(g)$ of non-Abelian groups of order $\le g$.
	\begin{figure}[h]
		\begin{center}
		\includegraphics[scale=0.8, keepaspectratio=true]{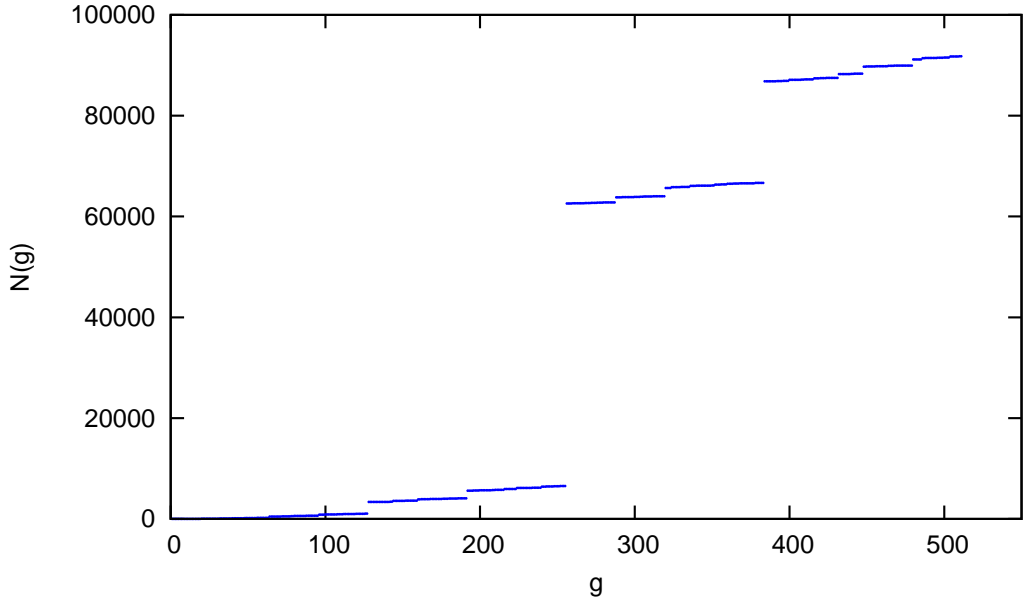}
		\caption{Total number $N(g)$ of non-Abelian groups up to order $g$.}\label{numbergroups}
		\end{center}
		\end{figure}
\hspace{0mm}
\\
From figure \ref{numbergroups} one can immediately deduce that the number of finite groups of order $g$ is usually very high if $g$ contains high powers of 2. Therefore there are high ``jumps'' in $N(g)$ at
	\begin{displaymath}
	g=2^8=256\quad\mbox{and}\quad g=3\times 2^7=384.
	\end{displaymath}
Indeed a much larger ``jump'' occurs at $g=512$: $N(511)=91774$, while there are $10494213$ groups of order $512$ \cite{SmallGroups} of which only $30$ are Abelian\footnote{The number of non-isomorphic Abelian groups of a given order can be calculated explicitly. See for example the article ``Abelian Group'' in \cite{weisstein}.}. If we want to analyse groups with faithful 3-dimensional irreducible representations only, we don't need to consider groups of order 512 due to the following theorem:
\begin{theorem}\label{dimensiontheorem}
Let $D$ be an irreducible representation of a finite group $G$, then the dimension $\mathrm{dim}(D)$ of $D$ is a divisor of the order $\mathrm{ord}(G)$ of $G$.
\end{theorem}
\hspace{0mm}
\\
The proof of this theorem can be found in textbooks on finite group theory, see for example \cite{speiser} p. 176f. or \cite{hall} p. 288f. Note that theorem \ref{dimensiontheorem} tells us that the order of any group which possesses a 3-dimensional irreducible representation must be divisible by 3. This implies that the groups of order $512$ do not possess 3-dimensional irreducible representations (but there could be groups of order 512 which possess faithful 3-dimensional reducible representations).
\\
In this work we will analyse all groups of order up to 511. From tables \ref{SU3-subgroups} and \ref{U3-subgroups}, which show our results, one can find that the orders of all groups we have found are indeed divisible by 3.

\subsection{Extraction of finite subgroups of $U(3)$ from the SmallGroups Library}

Using GAP the determination of the finite subgroups of $U(3)$ from the SmallGroups Library is not difficult. GAP offers the opportunity to calculate the character tables\footnote{We use the GAP command \texttt{CharacterTable}(.) to calculate the character table of a group.} as well as all irreducible representations of a given ``small group'' $\mbox{\textlbrackdbl}g,j\mbox{\textrbrackdbl}$. Using criterion \ref{criterion1} one can immediately deduce the dimensions of the faithful irreducible representations of a ``small group'' using its character table. If the analysed group has a three-dimensional faithful irreducible representation it is a finite subgroup of $U(3)$\footnote{Please note that this is sufficient, but not necessary \cite{parattu-wingerter}. Here we specialise onto the finite subgroups of $U(3)$ which possess a faithful 3-dimensional irreducible representation.}. By explicit construction of the irreducible representations\footnote{\label{footnote6}We use the GAP command \texttt{IrreducibleRepresentations}(.) to calculate the irreducible representations of a group. Since the labeling of the irreducible representations computed with \texttt{IrreducibleRepresentations}(.) does not necessarily agree with the labeling of \texttt{CharacterTable}(.), we use the commands \texttt{Image}(.) and \texttt{Order}(.) to find the faithful irreducible representations of the group under consideration.} one can determine the $U(3)$-subgroups which have a faithful three-dimensional irreducible representation of determinant 1 ($SU(3)$-subgroups).

\begin{crit}\label{criterion1}
Let $D$ be a $d$-dimensional representation of a finite group $G$. Then $D$ is non-faithful if and only if $D$ has more than one character $d$ in the character table.
\end{crit}
\hspace{0mm}
\\
The proof of criterion \ref{criterion1} can be found in appendix \ref{appendixA}.
\bigskip
\\
Let us, in this paper, choose the following convention: Let $G$ be a finite group. We say that ``$G$ can not be written as a direct product with a cyclic group'' if there does not exist a group $F$ and an $m>1$ such that
	\begin{equation}
	G\cong F\times\zed_m.
	\end{equation}
Before we list the results let us finally divide the obtained groups into another two sets, namely
	\begin{itemize}
	 \item groups that can be written as direct products with cyclic groups and
	 \item groups that can not be written as direct products with cyclic groups.
	\end{itemize}
How can we determine whether a ``small group'' can be written as a direct product with a cyclic group? The GAP command \texttt{StructureDescription}(.) gives the basic structure of a group, especially it tells us whether a group can be written as a direct product with a cyclic group. Let us clarify this with two examples:
	
	\begin{quote}
	\texttt{gap\textgreater StructureDescription(SmallGroup([12,3]));\\
	"A4" \\
	gap\textgreater StructureDescription(SmallGroup([24,13]));\\
	"C2 x A4"}
	\end{quote}
So GAP tells us that the group $\mbox{\textlbrackdbl}12,3\mbox{\textrbrackdbl}$ is isomorphic to $A_4$, and that $\mbox{\textlbrackdbl}24,13\mbox{\textrbrackdbl}$ is isomorphic to $\mathcal{Z}_2\times A_4$, i.e. it is a direct product with a cyclic group.
\medskip
\\
How are the groups that can be written as direct products with cyclic groups related to the groups which can not be written as direct products with cyclic groups? The answer is provided by theorem \ref{theorem1}.

\begin{theorem}\label{theorem1}
Let $G$ be a finite group with an $m$-dimensional faithful irreducible representation $D$. Let $C$ be the center of $G$, $\mathrm{ord}(C)=c$ and let $\mathrm{gcd}(n,c)$ be the greatest common divisor of $n,c\in \mathbbm{N}\backslash\{0\}$.
\medskip
\\
Then $\mathcal{Z}_n\times G$ has a faithful $m$-dimensional irreducible representation if and only if $\mathrm{gcd}(n,c)=1$.
\end{theorem}
\hspace{0mm}
\\
The proof of this theorem can be found in appendix \ref{appendixB}. Theorem \ref{theorem1} implies that we can construct all finite groups which have a faithful three-dimensional irreducible representation from all finite groups which have a faithful three-dimensional irreducible representation and can not be written as direct products with cyclic groups.
\medskip
\\
Let us consider the following examples:
	\begin{enumerate}
	 \item The group $A_4\cong \mbox{\textlbrackdbl}12,3\mbox{\textrbrackdbl}$ has center $C=\{e\}.\Rightarrow c=\mathrm{ord}(C)=1$, thus $n\in\mathbbm{N}\backslash\{0\}$ and $c=1$ have no common divisor $d\neq 1$. Therefore all direct products
		\begin{displaymath}
		\mathcal{Z}_n\times A_4,\quad n\in\mathbbm{N}\backslash\{0,1\}
		\end{displaymath}
	have three-dimensional faithful irreducible representations. Among these direct products only $\mathcal{Z}_3\times A_4$ will have a faithful three-dimensional irreducible representation of determinant 1 ($\mathrm{det}(\omega^k \mathbbm{1}_3)=1,\enspace \omega=e^{\frac{2\pi i}{3}},\enspace k=0,1,2$).
	 \item The group $\Delta(27)\cong \mbox{\textlbrackdbl}27,3\mbox{\textrbrackdbl}$ has center $C\cong \mathcal{Z}_3.\Rightarrow c=\mathrm{ord}(C)=3$, thus $n\in\mathbbm{N}\backslash\{3k|k\in\mathbbm{N}\}$ and $c=3$ have no common divisor $d\neq 1$. Therefore all direct products
		\begin{displaymath}
		\mathcal{Z}_n\times \Delta(27),\quad n\in\mathbbm{N}\backslash(\{3k|k\in\mathbbm{N}\}\cup\{1\})
		\end{displaymath}
	have faithful three-dimensional irreducible representations. None of these groups has a three-dimensional faithful irreducible representation of determinant 1.
	\end{enumerate}
In the results we will only list groups that can not be written as direct products with cyclic groups. From these groups all other groups can be constructed using theorem \ref{theorem1}. The results obtained from the SmallGroups Library are in perfect agreement with theorem \ref{theorem1}.

\subsection{Results}

\subsubsection{Generators}
In tables \ref{SU3generators} and \ref{U3generators} we list all matrices needed to generate the finite subgroups of $U(3)$ of order smaller than 512.

\begin{center}
\begin{longtable}{|ll|}
\firsthline
\multicolumn{2}{|c|}{Generators of determinant 1}\\
\hline
\endhead
\hline
\endfoot
\endlastfoot
\hline
\footnotesize
$\displaystyle{E = \left( \begin{array}{ccc}
	0 & 1 & 0 \\ 0 & 0 & 1 \\ 1 & 0 & 0
	\end{array} \right)}$ & \footnotesize $\displaystyle{F(n,a,b) = \left( \begin{array}{ccc}
	\eta^a & 0 & 0 \\ 0 & \eta^b & 0 \\ 0 & 0 & \eta^{-a-b}
	\end{array} \right)}$ \\ \footnotesize $\displaystyle{G(d,r,s) = \left( \begin{array}{ccc}
	\delta^r & 0 & 0 \\ 0 & 0 & \delta^s \\ 0 & -\delta^{-r-s} & 0
	\end{array} \right)}$ & \footnotesize
	$\displaystyle{H = \frac{1}{2} \left( \begin{array}{ccc}
	-1 & \mu_- & \mu_+ \\ \mu_- & \mu_+ & -1 \\ \mu_+ & -1 & \mu_-
	\end{array} \right)}$ \\ \footnotesize $\displaystyle{J = \left( \begin{array}{ccc}
	1 & 0 & 0 \\ 0 & \omega & 0 \\ 0 & 0 & \omega^2 
	\end{array} \right)}$ & \footnotesize $\displaystyle{K= 
	\frac{1}{\sqrt{3}\,i} \left( \begin{array}{ccc}
	1 & 1 & 1 \\ 
	1 & \omega & \omega^2 \\ 
	1 & \omega^2 & \omega
	\end{array} \right)}$\\ \footnotesize
	$\displaystyle{L = 
	\frac{1}{\sqrt{3}\,i} \left( \begin{array}{ccc}
	1 & 1 & \omega^2 \\ 
	1 & \omega & \omega \\ 
	\omega & 1 & \omega
	\end{array} \right)}$ & \footnotesize $\displaystyle{M= \left( \begin{array}{ccc}
	\beta & 0 & 0 \\ 0 & \beta^2 & 0 \\ 0 & 0 & \beta^4 
	\end{array} \right)}$ \\ \footnotesize $\displaystyle{N= 
	\frac{i}{\sqrt{7}} \left( \begin{array}{ccc}
	\beta^4 - \beta^3 & \beta^2 - \beta^5 & \beta   - \beta^6 \\ 
	\beta^2 - \beta^5 & \beta   - \beta^6 & \beta^4 - \beta^3 \\ 
	\beta   - \beta^6 & \beta^4 - \beta^3 & \beta^2 - \beta^5 
	\end{array} \right)}$ & \footnotesize
	$\displaystyle{P= 
	\left( \begin{array}{ccc}
	\epsilon & 0 & 0 \\ 
	0 & \epsilon & 0 \\ 
	0 & 0 & \epsilon\omega
	\end{array} \right)}$ \\
	\footnotesize $\displaystyle{Q=\left( \begin{array}{ccc}
	-1 & 0 & 0 \\ 
	0 & 0 & -\omega \\ 
	0 & -\omega^2 & 0
	\end{array} \right)}$ &
\\* 
\hline
\caption{Generators of finite subgroups of $SU(3)$.\\ $\eta:=e^{2\pi i/n},\quad \delta:=e^{2\pi i/d},\quad \mu_{\pm} = \frac{1}{2} \left( -1 \pm \sqrt{5} \right),\quad \omega=e^{2\pi i/3},\quad\beta=e^{2\pi i/7},\quad \epsilon=e^{4\pi i/9}$.}
\label{SU3generators}
\end{longtable}
\end{center}

\begin{center}
\begin{longtable}{|ll|}
\firsthline
\multicolumn{2}{|c|}{Generators for groups of determinant unequal 1}\\
\hline
\endhead
\hline
\endfoot
\endlastfoot
\hline
\footnotesize $\displaystyle{R(n,a,b,c) = \left( \begin{array}{ccc}
	0 & 0 & \eta^a \\ \eta^b & 0 & 0 \\ 0 & \eta^c & 0
	\end{array} \right)}$ & \footnotesize $\displaystyle{S(n,a,b,c) = \left( \begin{array}{ccc}
	\eta^a & 0 & 0 \\ 0 & 0 & \eta^b \\ 0 & \eta^c & 0
	\end{array} \right)}$\\ \footnotesize $\displaystyle{T(n,a,b,c) = \left( \begin{array}{ccc}
	0 & 0 & \eta^a \\ 0 & \eta^b & 0 \\ \eta^c & 0 & 0
	\end{array} \right)}$ &
	\footnotesize $\displaystyle{U(n,a,b,c) = \left( \begin{array}{ccc}
	0 & \eta^a & 0 \\ \eta^b & 0 & 0 \\ 0 & 0 & \eta^c
	\end{array} \right)}$ \\ \footnotesize $\displaystyle{V(n,a,b,c) = \left( \begin{array}{ccc}
	0 & \eta^a & 0 \\ 0 & 0 & \eta^b \\ \eta^c & 0 & 0 
	\end{array} \right)}$ & \footnotesize $\displaystyle{W(n,a,b,c) = \left( \begin{array}{ccc}
	\eta^a & 0 & 0 \\ 
	0 & \eta^b & 0 \\ 
	0 & 0 & \eta^c
	\end{array} \right)}$\\
	\footnotesize $\displaystyle{X_1=\left( \begin{array}{ccc}
	0 & \frac{1}{\sqrt{2}}\gamma^{11} & \frac{1}{\sqrt{2}}\gamma^{14} \\ 
	\frac{1}{\sqrt{2}}\gamma^5 & \frac{1}{2}\gamma^{20} & \frac{1}{2}\gamma^{11} \\ 
	\frac{1}{\sqrt{2}}\gamma^{14} & \frac{1}{2}\gamma^{17} & \frac{1}{2}\gamma^8
	\end{array} \right)}$ & \footnotesize $\displaystyle{X_2=\left( \begin{array}{ccc}
	\frac{1}{\sqrt{3}}\gamma^{21} & \frac{1}{\sqrt{6}}\gamma^{16} & \frac{1}{\sqrt{2}}\gamma^{13} \\ 
	\sqrt{\frac{2}{3}}\gamma^{14} & \frac{1}{2\sqrt{3}}\gamma^{21} & \frac{1}{2}\gamma^{18} \\ 
	0 & \frac{\sqrt{3}}{2}\gamma^{18} & \frac{1}{2}\gamma^3
	\end{array} \right)}$ \\ \footnotesize $\displaystyle{X_3=\left( \begin{array}{ccc}
	\frac{1}{\sqrt{3}}\vartheta^{31} & \frac{1}{\sqrt{6}}\vartheta^{14} & \frac{1}{\sqrt{2}}\vartheta^{4} \\ 
	\sqrt{\frac{2}{3}}\vartheta^{30} & \frac{1}{2\sqrt{3}}\vartheta^{31} & \frac{1}{2}\vartheta^{21} \\ 
	0 & \frac{\sqrt{3}}{2}\vartheta^{32} & \frac{1}{2}\vartheta^4
	\end{array} \right)}$ &
	\footnotesize $\displaystyle{X_4=\left( \begin{array}{ccc}
	0 & \frac{1}{\sqrt{2}}\vartheta^{13} & \frac{1}{\sqrt{2}}\vartheta^{12} \\ 
	\frac{1}{\sqrt{2}}\vartheta^{35} & \frac{1}{2}\vartheta^{24} & \frac{1}{2}\vartheta^{5} \\ 
	\frac{1}{\sqrt{2}}\vartheta^{18} & \frac{1}{2}\vartheta^{25} & \frac{1}{2}\vartheta^6
	\end{array} \right)}$ \\ \footnotesize $\displaystyle{X_5=\left( \begin{array}{ccc}
	\frac{1}{\sqrt{3}}\phi^9 & \sqrt{\frac{2}{3}} & 0 \\ 
	\sqrt{\frac{2}{3}}\phi^2 & \frac{1}{\sqrt{3}}\phi & 0 \\ 
	0 & 0 & \phi^5
	\end{array} \right)}$ & \footnotesize $\displaystyle{X_6=\left( \begin{array}{ccc}
	\gamma^{22} & 0 & 0 \\ 
	0 & \frac{1}{2}\gamma^{10} & \frac{\sqrt{3}}{2}\gamma^{11} \\ 
	0 & \frac{\sqrt{3}}{2}\gamma^{21} & \frac{1}{2}\gamma^{10}
	\end{array} \right)}$\\
	\footnotesize $\displaystyle{X_7=\left( \begin{array}{ccc}
	\frac{1}{\sqrt{3}}\psi^{9} & \frac{1}{\sqrt{6}}\psi^{2} & \frac{1}{\sqrt{2}}\psi^7 \\ 
	\frac{1}{\sqrt{6}}\psi^{4} & \frac{9+\sqrt{3}i}{12} & \frac{1}{2}\psi^{10} \\ 
	\frac{1}{\sqrt{2}}\psi^{11} & \frac{1}{2} & \frac{1}{2}\psi^2
	\end{array} \right)}$ & \footnotesize $\displaystyle{X_8=\left( \begin{array}{ccc}
	\frac{1}{\sqrt{3}}\psi^6 & \sqrt{\frac{2}{3}}\psi & 0 \\ 
	\sqrt{\frac{2}{3}}\psi^{11} & \frac{1}{\sqrt{3}} & 0 \\ 
	0 & 0 & \psi^3
	\end{array} \right)}$ \\ \footnotesize $\displaystyle{X_9=\left( \begin{array}{ccc}
	\frac{1}{\sqrt{3}}\gamma^{13} & \sqrt{\frac{2}{3}}\gamma^{14} & 0 \\ 
	\sqrt{\frac{2}{3}}\gamma^{12} & \frac{1}{\sqrt{3}}\gamma & 0 \\ 
	0 & 0 & \gamma^{19}
	\end{array} \right)}$ &
	\footnotesize $\displaystyle{X_{10}=\left( \begin{array}{ccc}
	0 & \frac{1}{\sqrt{2}}\gamma^3 & \frac{1}{\sqrt{2}}\gamma^{19} \\ 
	\frac{1}{\sqrt{2}}\gamma & \frac{1}{2}\gamma^2 & \frac{1}{2}\gamma^6 \\ 
	\frac{1}{\sqrt{2}}\gamma^{21} & \frac{1}{2}\gamma^{10} & \frac{1}{2}\gamma^{14}
	\end{array} \right)}$
\\* 
\hline
\caption{Generators of finite subgroups of $U(3)$.\\ $\gamma:=e^{2\pi i/24}$, $\vartheta=e^{2\pi i/36}$, $\phi=e^{2\pi i/16}$, $\psi=e^{2\pi i/12}$.}
\label{U3generators}
\end{longtable}
\end{center}
\hspace{0mm}
\\
Following \cite{miller,finitesubgroups-su3-fairbairn} all non-Abelian finite subgroups of $SU(3)$ which have a faithful three-dimensional irreducible representation can be cast into one of the types\footnote{Note that for some choices of $n,a,b,d,r,s$ the three-dimensional representations of $C(n,a,b)$, $D(n,a,b;d,r,s)$ given here could be reducible or lead to direct products with cyclic groups, so not all values of the parameters are allowed.\\
The allowed values for $n$ in $T_n$ are products of powers of primes of the form $3k+1$, $k\in\mathbbm{N}$. Please note furthermore that $T_n$ is in general not unique. The equation $(1+a+a^2)\hspace{1mm}\mathrm{mod}\hspace{1mm}n=0$ may have more than one solution, which can lead to non-isomorphic groups $T_n$ with the same $n$. There are for example two non-isomorphic groups $T_{91}$ in table \ref{SU3-subgroups}.} listed in table \ref{types-SU3-subgroups}.

\begin{table}
\begin{center}
\begin{tabular}{|lll|}
\hline
Group & Generators & References \\
\hline
$C(n,a,b)$ & $E,\enspace F(n,a,b)$ & \cite{miller,DTPOL} \\ 
$D(n,a,b;d,r,s)$ & $E,\enspace F(n,a,b),\enspace G(d,r,s)$ & \cite{miller,DTPOL,Zwicky}\\
$\Delta(3n^2)=C(n,0,1),\enspace n\ge 2$ & $E,\enspace F(n,0,1)$ & \cite{finitesubgroups-su3-fairbairn,BLW1,D3n^2,DTPOL}\\
$\Delta(6n^2)=D(n,0,1;2,1,1),\enspace n\ge 2$ & $E,\enspace F(n,0,1),\enspace G(2,1,1)$ & \cite{finitesubgroups-su3-fairbairn,BLW1,D6n^2,DTPOL}\\
$T_n=C(n,1,a),\enspace (1+a+a^2)\hspace{1mm}\mathrm{mod}\hspace{1mm}n=0$ & $E,\enspace F(n,1,a)$ & \cite{BLW1,BLW2,Fairbairn2,DTPOL}\\
$A_5=\Sigma(60)$ & $E,\enspace F(2,0,1),\enspace H$ & \cite{miller,finitesubgroups-su3-fairbairn,SimpleFinite,DTPOL,A5}\\
$PSL(2,7)=\Sigma(168)$ & $E,\enspace M,\enspace N$ & \cite{miller,finitesubgroups-su3-fairbairn,SimpleFinite,DTPOL}\\
$\Sigma(36\phi)$ & $E,\enspace J, \enspace K$ & \cite{miller,finitesubgroups-su3-fairbairn,DTPOL,Principal_series}\\
$\Sigma(72\phi)$ & $E,\enspace J, \enspace K,\enspace L$ & \cite{miller,finitesubgroups-su3-fairbairn,DTPOL,Principal_series}\\
$\Sigma(216\phi)$ & $E,\enspace J, \enspace K,\enspace P$ & \cite{miller,finitesubgroups-su3-fairbairn,DTPOL,Principal_series}\\
$\Sigma(360\phi)$ & $E,\enspace F(2,0,1),\enspace H,\enspace Q$ & \cite{miller,finitesubgroups-su3-fairbairn,DTPOL}
\\
\hline
\end{tabular}
\caption{Types of finite subgroups of $SU(3)$ \cite{miller,finitesubgroups-su3-fairbairn}.}
\label{types-SU3-subgroups}
\end{center}
\end{table}
\hspace{0mm}\\
To our knowledge only one series of finite subgroups of $U(3)$ (determinant $\neq$ 1) is known by now, namely $\Sigma(3N^3),\enspace (N\in\{3k\vert k\in\mathbbm{N}\backslash\{0,1\}\})$, which has been published recently by Ishimori et al. in \cite{Ishimori:2010au}. A well known member of $\Sigma(3N^3)$ is $\Sigma(81)$ \cite{T7andSigma81,Sigma81}.

\subsubsection{The finite subgroups of $SU(3)$ of order smaller than 512}\label{SU3-results}

In table \ref{SU3-subgroups} we list the finite groups of order smaller than 512 that can not be written as direct products with cyclic groups and that have a faithful three-dimensional irreducible representation of determinant 1.

\begin{center}
\begin{longtable}{|lllcl|}
\firsthline
$\mbox{\textlbrackdbl}g,j\mbox{\textrbrackdbl}$ & Classification & Other names & $\mathrm{ord}(C)$ & References\\
\hline
\endhead
\hline
\endfoot
\endlastfoot
\hline
$\mbox{\textlbrackdbl} 12, 3 \mbox{\textrbrackdbl}$ & $\Delta(12)=\Delta(3\times 2^2)$ & $A_4,\enspace T$ & 1 & \cite{miller,DTPOL,Hamermesh,A4}\\
$\mbox{\textlbrackdbl} 21, 1 \mbox{\textrbrackdbl}$ & $C(7,1,2)$ & $T_7$ & 1 & \cite{BLW1,BLW2,Fairbairn2,DTPOL,T7,T7andSigma81}\\
$\mbox{\textlbrackdbl} 24, 12 \mbox{\textrbrackdbl}$ & $\Delta(24)=\Delta(6\times 2^2)$ & $S_4,\enspace O$ & 1 & \cite{miller,DTPOL,Hamermesh,S4} \\
$\mbox{\textlbrackdbl} 27, 3 \mbox{\textrbrackdbl}$ & $\Delta(27)=\Delta(3\times 3^2)$ & & 3 & \cite{Delta(27)a,Delta(27)b}\\
$\mbox{\textlbrackdbl} 39, 1 \mbox{\textrbrackdbl}$ & $C(13,1,3)$ & $T_{13}$ & 1 & \\
$\mbox{\textlbrackdbl} 48, 3 \mbox{\textrbrackdbl}$ & $\Delta(48)=\Delta(3\times 4^2)$ & & 1 & \\
$\mbox{\textlbrackdbl} 54, 8 \mbox{\textrbrackdbl}$ & $\Delta(54)=\Delta(6\times 3^2)$ & & 3 & \cite{Delta(54)}\\
$\mbox{\textlbrackdbl} 57, 1 \mbox{\textrbrackdbl}$ & $C(19,1,7)$ & $T_{19}$ & 1 & \\
$\mbox{\textlbrackdbl} 60, 5 \mbox{\textrbrackdbl}$ & $A_5$ & $\Sigma(60),\enspace I$ & 1 & \cite{miller,finitesubgroups-su3-fairbairn,SimpleFinite,DTPOL,A5,Hamermesh} \\
$\mbox{\textlbrackdbl} 75, 2 \mbox{\textrbrackdbl}$ & $\Delta(75)=\Delta(3\times 5^2)$ & & 1 & \\
$\mbox{\textlbrackdbl} 81, 9 \mbox{\textrbrackdbl}$ & $C(9,1,1)$ & & 3 & \\
$\mbox{\textlbrackdbl} 84, 11 \mbox{\textrbrackdbl}$ & $C(14,1,2)$ & & 1 & \\
$\mbox{\textlbrackdbl} 93, 1 \mbox{\textrbrackdbl}$ & $C(31,1,5)$ & $T_{31}$ & 1 & \\
$\mbox{\textlbrackdbl} 96, 64 \mbox{\textrbrackdbl}$ & $\Delta(96)=\Delta(6\times 4^2)$ & & 1 & \\
$\mbox{\textlbrackdbl} 108, 15 \mbox{\textrbrackdbl}$ & $\Sigma(36\phi)$ & & 3 & \cite{miller,finitesubgroups-su3-fairbairn,DTPOL,Principal_series}\\
$\mbox{\textlbrackdbl} 108, 22 \mbox{\textrbrackdbl}$ & $\Delta(108)=\Delta(3\times 6^2)$ & & 3 & \\
$\mbox{\textlbrackdbl} 111, 1 \mbox{\textrbrackdbl}$ & $C(37,1,10)$ & $T_{37}$ & 1 & \\
$\mbox{\textlbrackdbl} 129, 1 \mbox{\textrbrackdbl}$ & $C(43,1,6)$ & $T_{43}$ & 1 & \\
$\mbox{\textlbrackdbl} 147, 1 \mbox{\textrbrackdbl}$ & $C(49,10,6)$ & $T_{49}$ & 1 & \\
$\mbox{\textlbrackdbl} 147, 5 \mbox{\textrbrackdbl}$ & $\Delta(147)=\Delta(3\times 7^2)$ & & 1 & \\
$\mbox{\textlbrackdbl} 150, 5 \mbox{\textrbrackdbl}$ & $\Delta(150)=\Delta(6\times 5^2)$ & & 1 & \\
$\mbox{\textlbrackdbl} 156, 14 \mbox{\textrbrackdbl}$ & $C(26,1,3)$ & & 1 & \\
$\mbox{\textlbrackdbl} 162, 14 \mbox{\textrbrackdbl}$ & $D(9,1,1;2,1,1)$ & & 3 &\\
$\mbox{\textlbrackdbl} 168, 42 \mbox{\textrbrackdbl}$ & $PSL(2,7)$ & $\Sigma(168)$ & 1 & \cite{miller,SimpleFinite,DTPOL,PSL27} \\
$\mbox{\textlbrackdbl} 183, 1 \mbox{\textrbrackdbl}$ & $C(61,1,13)$ & $T_{61}$ & 1 & \\
$\mbox{\textlbrackdbl} 189, 8 \mbox{\textrbrackdbl}$ & $C(21,1,2)$ & & 3 & \\
$\mbox{\textlbrackdbl} 192, 3 \mbox{\textrbrackdbl}$ & $\Delta(192)=\Delta(3\times 8^2)$ & & 1 & \\
$\mbox{\textlbrackdbl} 201, 1 \mbox{\textrbrackdbl}$ & $C(67,1,29)$ & $T_{67}$ & 1 & \\
$\mbox{\textlbrackdbl} 216, 88 \mbox{\textrbrackdbl}$ & $\Sigma(72\phi)$ & & 3 & \cite{miller,finitesubgroups-su3-fairbairn,DTPOL,Principal_series}\\
$\mbox{\textlbrackdbl} 216, 95 \mbox{\textrbrackdbl}$ & $\Delta(216)=\Delta(6\times 6^2)$ & & 3 & \\
$\mbox{\textlbrackdbl} 219, 1 \mbox{\textrbrackdbl}$ & $C(73,1,8)$ & $T_{73}$ & 1 & \\
$\mbox{\textlbrackdbl} 228, 11 \mbox{\textrbrackdbl}$ & $C(38,1,7)$ & & 1 & \\
$\mbox{\textlbrackdbl} 237, 1 \mbox{\textrbrackdbl}$ & $C(79,1,23)$ & $T_{79}$ & 1 & \\
$\mbox{\textlbrackdbl} 243, 26 \mbox{\textrbrackdbl}$ & $\Delta(243)=\Delta(3\times 9^2)$ & & 3 & \\
$\mbox{\textlbrackdbl} 273, 3 \mbox{\textrbrackdbl}$ & $C(91,1,16)$ & $T_{91}$ & 1 & \\
$\mbox{\textlbrackdbl} 273, 4 \mbox{\textrbrackdbl}$ & $C(91,1,9)$ & $T_{91}$ & 1 & \\
$\mbox{\textlbrackdbl} 291, 1 \mbox{\textrbrackdbl}$ & $C(97,1,35)$ & $T_{97}$ & 1 & \\
$\mbox{\textlbrackdbl} 294, 7 \mbox{\textrbrackdbl}$ & $\Delta(294)=\Delta(6\times 7^2)$ & & 1 & \\
$\mbox{\textlbrackdbl} 300, 43 \mbox{\textrbrackdbl}$ & $\Delta(300)=\Delta(3\times 10^2)$ & & 1 & \\
$\mbox{\textlbrackdbl} 309, 1 \mbox{\textrbrackdbl}$ & $C(103,1,46)$ & $T_{103}$ & 1 & \\
$\mbox{\textlbrackdbl} 324, 50 \mbox{\textrbrackdbl}$ & $C(18,1,1)$ & & 3 & \\
$\mbox{\textlbrackdbl} 327, 1 \mbox{\textrbrackdbl}$ & $C(109,1,45)$ & $T_{109}$ & 1 & \\
$\mbox{\textlbrackdbl} 336, 57 \mbox{\textrbrackdbl}$ & $C(28,1,2)$ & & 1 & \\
$\mbox{\textlbrackdbl} 351, 8 \mbox{\textrbrackdbl}$ & $C(39,1,3)$ & & 3 & \\
$\mbox{\textlbrackdbl} 363, 2 \mbox{\textrbrackdbl}$ & $\Delta(363)=\Delta(3\times 11^2)$ & & 1 & \\
$\mbox{\textlbrackdbl} 372, 11 \mbox{\textrbrackdbl}$ & $C(62,1,5)$ & & 1 & \\
$\mbox{\textlbrackdbl} 381, 1 \mbox{\textrbrackdbl}$ & $C(127,1,19)$ & $T_{127}$ & 1 & \\
$\mbox{\textlbrackdbl} 384, 568 \mbox{\textrbrackdbl}$ & $\Delta(384)=\Delta(6\times 8^2)$ & & 1 & \\
$\mbox{\textlbrackdbl} 399, 3 \mbox{\textrbrackdbl}$ & $C(133,1,11)$ & $T_{133}$ & 1 & \\
$\mbox{\textlbrackdbl} 399, 4 \mbox{\textrbrackdbl}$ & $C(133,1,30)$ & $T_{133}$ & 1 & \\
$\mbox{\textlbrackdbl} 417, 1 \mbox{\textrbrackdbl}$ & $C(139,1,42)$ & $T_{139}$ & 1 & \\
$\mbox{\textlbrackdbl} 432, 103 \mbox{\textrbrackdbl}$ & $\Delta(432)=\Delta(3\times 12^2)$ & & 3 &\\
$\mbox{\textlbrackdbl} 444, 14 \mbox{\textrbrackdbl}$ & $C(74,1,10)$ & & 1 & \\
$\mbox{\textlbrackdbl} 453, 1 \mbox{\textrbrackdbl}$ & $C(151,1,32)$ & $T_{151}$ & 1 & \\
$\mbox{\textlbrackdbl} 471, 1 \mbox{\textrbrackdbl}$ & $C(157,1,12)$ & & 1 & \\
$\mbox{\textlbrackdbl} 486, 61 \mbox{\textrbrackdbl}$ & $\Delta(486)=\Delta(6\times 9^2)$ & & 3 & \\
$\mbox{\textlbrackdbl} 489, 1 \mbox{\textrbrackdbl}$ & $C(163,1,58)$ & $T_{163}$ & 1 & \\
$\mbox{\textlbrackdbl} 507, 1 \mbox{\textrbrackdbl}$ & $C(169,1,22)$ & $T_{169}$ & 1 & \\
$\mbox{\textlbrackdbl} 507, 5 \mbox{\textrbrackdbl}$ & $\Delta(507)=\Delta(3\times 13^2)$ & & 1 &
\\* 
\hline
\caption{The finite subgroups of $SU(3)$ of order smaller than 512. $\mbox{\textlbrackdbl}g,j\mbox{\textrbrackdbl}$ denotes the SmallGroups number and $\mathrm{ord}(C)$ denotes the order of the center of the group.}
\label{SU3-subgroups}
\end{longtable}
\end{center}

\subsubsection{The finite subgroups of $U(3)$ of order smaller than 512}\label{U3-results}

In table \ref{U3-subgroups} we list the finite groups of order smaller than 512 that can not be written as direct products with cyclic groups and that have a faithful three-dimensional irreducible representation of determinant unequal 1.
\medskip
\\
The generators of these groups were taken from the list of faithful three-dimensional irreducible representations constructed with GAP (see footnote 6 on page \pageref{footnote6}). For most groups GAP constructed unitary representations. For the groups where GAP did not give a unitary representation we constructed a unitary representation in the following way: Let $D$ be an $n$-dimensional representation of a group $G$ and let $\{v_j|j=1,...,n\}$ be an orthonormal basis of $\mathbbm{C}^n$ with respect to the scalar product
	\begin{equation}
	\langle x,y\rangle:=\frac{1}{\mathrm{ord}(G)}\sum_{a\in G}(D(a)x,D(a)y),
	\end{equation}
where $(x,y):=x^{\dagger} y$ is the standard scalar product on $\mathbbm{C}^n$. Then if we define $T:=(v_1,...,v_n)$, the representation $T^{-1}DT$ is unitary with respect to $(.\hspace{0.5mm},.)$ . Usually this construction is used to prove that every representation of a finite group is equivalent to a unitary representation. The power of modern computer algebra systems allows us to explicitly calculate the scalar product $\langle .,.\rangle$ and to construct $T$ by Gram-Schmidt orthogonalisation. In this way the unitary generators $X_1,...,X_{10}$ were obtained.

\begin{center}
\begin{longtable}{|llll|}
\firsthline
$\mbox{\textlbrackdbl}g,j\mbox{\textrbrackdbl}$ & Classification & Generators & $\mathrm{ord}(C)$\\
\hline
\endhead
\hline
\endfoot
\endlastfoot
\hline
$\mbox{\textlbrackdbl} 27, 4 \mbox{\textrbrackdbl}$ & & $R(3,1,1,2),\enspace R(3,1,2,1)$ & 3\\
$\mbox{\textlbrackdbl} 36, 3 \mbox{\textrbrackdbl}$ & $\Delta(3\times 2^2,2)$ & $R(6,2,1,1),\enspace R(3,0,1,1)$ & 3\\
$\mbox{\textlbrackdbl} 48, 30 \mbox{\textrbrackdbl}$ & $S_4(2)$ & $S(4,1,3,1),\enspace T(4,3,3,1)$ & 2\\
$\mbox{\textlbrackdbl} 63, 1 \mbox{\textrbrackdbl}$ & $T_7(2)$ & $R(21,5,10,13),\enspace R(21,3,20,5)$ & 3\\
$\mbox{\textlbrackdbl} 81, 6 \mbox{\textrbrackdbl}$ & & $R(9,2,4,7),\enspace R(9,3,8,5)$ & 9\\
$\mbox{\textlbrackdbl} 81, 7 \mbox{\textrbrackdbl}$ & $\Sigma(3\times 3^3)$ & $R(3,2,2,0),\enspace R(3,1,1,0)$ & 3\\
$\mbox{\textlbrackdbl} 81, 8 \mbox{\textrbrackdbl}$ & & $R(9,2,2,8),\enspace R(9,4,4,7)$ & 3\\
$\mbox{\textlbrackdbl} 81, 10 \mbox{\textrbrackdbl}$ & & $R(9,4,7,4),\enspace R(9,2,5,8)$ & 3\\
$\mbox{\textlbrackdbl} 81, 14 \mbox{\textrbrackdbl}$ & & $R(9,4,7,1),\enspace R(9,8,5,2),\enspace R(9,6,3,0)$ & 9\\
$\mbox{\textlbrackdbl} 96, 65 \mbox{\textrbrackdbl}$ & $S_4(3)$ & $S(8,1,5,1),\enspace T(8,3,3,7)$ & 4\\
$\mbox{\textlbrackdbl} 108, 3 \mbox{\textrbrackdbl}$ & $\Delta(3\times 2^2,3)$ & $R(18,4,5,5),\enspace R(9,0,5,5)$ & 9\\
$\mbox{\textlbrackdbl} 108, 11 \mbox{\textrbrackdbl}$ & $\Delta(6\times 3^2,2)$ & $S(12,5,9,1),\enspace T(12,3,7,11),\enspace U(12,1,9,5)$ & 6\\
$\mbox{\textlbrackdbl} 108, 19 \mbox{\textrbrackdbl}$ & & $R(18,4,7,1),\enspace R(9,4,7,1)$ & 3\\
$\mbox{\textlbrackdbl} 108, 21 \mbox{\textrbrackdbl}$ & & $R(6,0,3,5),\enspace R(3,2,0,2)$ & 3\\
$\mbox{\textlbrackdbl} 117, 1 \mbox{\textrbrackdbl}$ & $T_{13}(2)$ & $R(39,1,29,35),\enspace R(39,15,19,31)$ & 3\\
$\mbox{\textlbrackdbl} 144, 3 \mbox{\textrbrackdbl}$ & $\Delta(3\times 4^2,2)$ & $R(12,7,8,5),\enspace R(12,6,4,10)$ & 3\\
$\mbox{\textlbrackdbl} 162, 10 \mbox{\textrbrackdbl}$ & & $S(3,0,1,0),\enspace T(3,1,1,0)$ & 3\\
$\mbox{\textlbrackdbl} 162, 12 \mbox{\textrbrackdbl}$ & & $S(9,4,7,4),\enspace T(9,2,2,8)$ & 3\\
$\mbox{\textlbrackdbl} 162, 44 \mbox{\textrbrackdbl}$ & $\Delta'(6\times 3^2,2,1)$ & $S(9,2,8,5),\enspace T(9,4,1,7),\enspace U(9,5,8,2)$ & 9\\
$\mbox{\textlbrackdbl} 171, 1 \mbox{\textrbrackdbl}$ & $T_{19}(2)$ & $R(57,7,11,20),\enspace R(57,33,22,40)$ & 3\\
$\mbox{\textlbrackdbl} 189, 1 \mbox{\textrbrackdbl}$ & $T_7(3)$ & $R(63,22,23,32),\enspace R(63,9,46,1)$ & 9\\
$\mbox{\textlbrackdbl} 189, 4 \mbox{\textrbrackdbl}$ & & $R(63,1,58,25),\enspace R(63,2,53,50)$ & 3\\
$\mbox{\textlbrackdbl} 189, 5 \mbox{\textrbrackdbl}$ & & $R(63,1,16,4),\enspace R(63,2,32,8)$ & 3\\
$\mbox{\textlbrackdbl} 189, 7 \mbox{\textrbrackdbl}$ & & $R(21,5,17,6),\enspace R(21,3,13,12)$ & 3\\
$\mbox{\textlbrackdbl} 192, 182 \mbox{\textrbrackdbl}$ & $\Delta(6\times 4^2,2)$ & $T(4,0,2,1),\enspace U(4,3,0,2)$ & 2\\
$\mbox{\textlbrackdbl} 192, 186 \mbox{\textrbrackdbl}$ & $S_4(4)$ & $S(16,1,9,1),\enspace T(16,3,3,11)$ & 8\\
$\mbox{\textlbrackdbl} 216, 17 \mbox{\textrbrackdbl}$ & $\Delta(6\times 3^2,3)$ & $T(24,3,11,19),\enspace T(24,21,13,5),\enspace U(24,17,9,1)$ & 12\\
$\mbox{\textlbrackdbl} 216, 25 \mbox{\textrbrackdbl}$ & & $X_1,\enspace X_2$ & 6\\
$\mbox{\textlbrackdbl} 225, 3 \mbox{\textrbrackdbl}$ & $\Delta(3\times 5^2,2)$ & $R(15,1,11,8),\enspace R(15,12,7,1)$ & 3\\
$\mbox{\textlbrackdbl} 243, 16 \mbox{\textrbrackdbl}$ & & $R(9,2,4,7),\enspace R(9,6,8,5)$ & 9\\
$\mbox{\textlbrackdbl} 243, 19 \mbox{\textrbrackdbl}$ & & $V(27,5,14,5),\enspace W(27,2,2,11)$ & 9\\
$\mbox{\textlbrackdbl} 243, 20 \mbox{\textrbrackdbl}$ & & $V(27,5,23,5),\enspace W(27,2,2,20)$ & 9\\
$\mbox{\textlbrackdbl} 243, 24 \mbox{\textrbrackdbl}$ & & $R(27,5,13,22),\enspace R(27,9,26,17)$ & 27\\
$\mbox{\textlbrackdbl} 243, 25 \mbox{\textrbrackdbl}$ & & $R(9,0,2,4),\enspace R(9,0,4,8)$ & 3\\
$\mbox{\textlbrackdbl} 243, 27 \mbox{\textrbrackdbl}$ & & $R(9,2,0,4),\enspace R(9,7,0,8)$ & 3\\
$\mbox{\textlbrackdbl} 243, 50 \mbox{\textrbrackdbl}$ & & $R(27,5,23,14),\enspace V(27,11,20,2),\enspace V(27,17,17,17)$ & 27\\
$\mbox{\textlbrackdbl} 243, 55 \mbox{\textrbrackdbl}$ & & $R(9,2,2,8),\enspace R(9,4,4,7),\enspace V(3,1,0,0)$ & 9\\
$\mbox{\textlbrackdbl} 252, 11 \mbox{\textrbrackdbl}$ & & $R(42,23,25,22),\enspace R(21,9,4,1)$ & 3\\
$\mbox{\textlbrackdbl} 279, 1 \mbox{\textrbrackdbl}$ & $T_{31}(2)$ & $R(93,7,20,35),\enspace R(93,45,40,70)$ & 3\\
$\mbox{\textlbrackdbl} 300, 13 \mbox{\textrbrackdbl}$ & $\Delta(6\times 5^2,2)$ & $S(20,1,9,5),\enspace T(20,15,19,11)$ & 2\\
$\mbox{\textlbrackdbl} 324, 3 \mbox{\textrbrackdbl}$ & $\Delta(3\times 2^2,4)$ & $R(54,10,17,17),\enspace R(27,0,17,17)$ & 27\\
$\mbox{\textlbrackdbl} 324, 13 \mbox{\textrbrackdbl}$ & & $S(12,3,7,3),\enspace T(12,1,1,9)$ & 6\\
$\mbox{\textlbrackdbl} 324, 15 \mbox{\textrbrackdbl}$ & & $S(36,1,13,1),\enspace T(36,23,23,11)$ & 6\\
$\mbox{\textlbrackdbl} 324, 17 \mbox{\textrbrackdbl}$ & & $S(36,1,25,1),\enspace T(36,35,35,11)$ & 6\\
$\mbox{\textlbrackdbl} 324, 43 \mbox{\textrbrackdbl}$ & & $R(54,10,19,1),\enspace R(54,20,38,2)$ & 9\\
$\mbox{\textlbrackdbl} 324, 45 \mbox{\textrbrackdbl}$ & & $R(18,4,17,5),\enspace R(9,3,8,5)$ & 9\\
$\mbox{\textlbrackdbl} 324, 49 \mbox{\textrbrackdbl}$ & & $R(18,4,13,7),\enspace R(9,4,4,7)$ & 3\\
$\mbox{\textlbrackdbl} 324, 51 \mbox{\textrbrackdbl}$ & & $R(18,4,13,13),\enspace R(9,7,4,4)$ & 3\\
$\mbox{\textlbrackdbl} 324, 60 \mbox{\textrbrackdbl}$ & & $R(6,0,3,5),\enspace R(3,0,0,2)$ & 3\\
$\mbox{\textlbrackdbl} 324, 102 \mbox{\textrbrackdbl}$ & $\Delta'(6\times 3^2,2,2)$ & $T(36,1,25,13),\enspace T(36,29,5,17),\enspace U(36,35,11,23)$ & 18\\
$\mbox{\textlbrackdbl} 324, 111 \mbox{\textrbrackdbl}$ & & $X_3,\enspace X_4$ & 9\\
$\mbox{\textlbrackdbl} 324, 128 \mbox{\textrbrackdbl}$ & & $R(18,4,7,1),\enspace R(9,4,7,1),\enspace R(18,8,17,17)$ & 9\\
$\mbox{\textlbrackdbl} 333, 1 \mbox{\textrbrackdbl}$ & $T_{37}(2)$ & $R(111,1,26,47),\enspace R(111,39,52,94)$ & 3\\
$\mbox{\textlbrackdbl} 351, 1 \mbox{\textrbrackdbl}$ & $T_{13}(3)$ & $R(117,8,37,46),\enspace R(117,81,74,92)$ & 9\\
$\mbox{\textlbrackdbl} 351, 4 \mbox{\textrbrackdbl}$ & & $R(117,16,100,40),\enspace R(117,32,83,80)$ & 3\\
$\mbox{\textlbrackdbl} 351, 5 \mbox{\textrbrackdbl}$ & & $R(117,16,22,1),\enspace R(117,32,44,2)$ & 3\\
$\mbox{\textlbrackdbl} 351, 7 \mbox{\textrbrackdbl}$ & & $R(39,1,16,9),\enspace R(39,15,32,18)$ & 3\\
$\mbox{\textlbrackdbl} 384, 571 \mbox{\textrbrackdbl}$ & $\Delta(6\times 4^2,3)$ & $T(8,1,5,3),\enspace U(8,1,3,7)$ & 4\\
$\mbox{\textlbrackdbl} 384, 581 \mbox{\textrbrackdbl}$ & $S_4(5)$ & $S(32,1,17,1),\enspace T(32,3,3,19)$ & 16\\
$\mbox{\textlbrackdbl} 387, 1 \mbox{\textrbrackdbl}$ & $T_{43}(2)$ & $R(129,11,52,109),\enspace R(129,108,104,89)$ & 3\\
$\mbox{\textlbrackdbl} 432, 3 \mbox{\textrbrackdbl}$ & $\Delta(3\times 4^2,3)$ & $R(36,1,8,35),\enspace R(36,18,16,34)$ & 9\\
$\mbox{\textlbrackdbl} 432, 33 \mbox{\textrbrackdbl}$ & $\Delta(6\times 3^2,4)$ & $T(48,3,19,35),\enspace U(48,25,9,41),\enspace U(48,29,29,29)$ & 24\\
$\mbox{\textlbrackdbl} 432, 57 \mbox{\textrbrackdbl}$ & & $X_5,\enspace X_6$ & 12\\
$\mbox{\textlbrackdbl} 432, 100 \mbox{\textrbrackdbl}$ & & $R(36,1,22,25),\enspace R(18,1,4,7)$ & 3\\
$\mbox{\textlbrackdbl} 432, 102 \mbox{\textrbrackdbl}$ & & $R(12,3,0,1),\enspace R(6,1,0,1)$ & 3\\
$\mbox{\textlbrackdbl} 432, 239 \mbox{\textrbrackdbl}$ & & $X_7,\enspace X_8$ & 6\\
$\mbox{\textlbrackdbl} 432, 260 \mbox{\textrbrackdbl}$ & $\Delta(6\times 4^2,2)$ & $S(12,7,5,3),\enspace T(12,9,5,7),\enspace U(12,9,1,11)$ & 6\\
$\mbox{\textlbrackdbl} 432, 273 \mbox{\textrbrackdbl}$ & & $X_9,\enspace X_{10}$ & 12\\
$\mbox{\textlbrackdbl} 441, 1 \mbox{\textrbrackdbl}$ & & $R(147,94,125,26),\enspace R(147,90,103,52)$ & 3\\
$\mbox{\textlbrackdbl} 441, 7 \mbox{\textrbrackdbl}$ & $\Delta(3\times 7^2,2)$ & $R(21,1,8,5),\enspace R(21,9,16,10)$ & 3\\
$\mbox{\textlbrackdbl} 468, 14 \mbox{\textrbrackdbl}$ & & $R(78,2,19,31),\enspace R(39,15,19,31)$ & 3\\
$\mbox{\textlbrackdbl} 486, 26 \mbox{\textrbrackdbl}$ & & $S(27,5,14,5),\enspace T(27,19,19,10)$ & 9\\
$\mbox{\textlbrackdbl} 486, 28 \mbox{\textrbrackdbl}$ & & $S(27,5,23,5),\enspace T(27,1,1,10)$ & 9\\
$\mbox{\textlbrackdbl} 486, 125 \mbox{\textrbrackdbl}$ & & $S(9,2,5,2),\enspace T(9,7,7,4),\enspace U(3,0,1,1)$ & 9\\
$\mbox{\textlbrackdbl} 486, 164 \mbox{\textrbrackdbl}$ & $\Delta'(6\times 3^2,3,1)$ & $T(27,5,14,23),\enspace T(27,7,25,16),\enspace U(27,19,10,1)$ & 27\\* 
\hline
\caption{The finite subgroups of $U(3)$ (which are not finite subgroups of $SU(3)$) of order smaller than 512.}
\label{U3-subgroups}
\end{longtable}
\end{center}

\subsubsection{Numerical consistency check of the obtained results.}

The results listed in sections \ref{SU3-results} and \ref{U3-results} are based on the computer algebra system GAP \cite{GAP} and the SmallGroups Library \cite{SmallGroups,SmallGroupsLibrary}. As already mentioned, our results are in perfect agreement with theorem \ref{theorem1}, which is the reason we did not list groups that can be written as direct products with cyclic groups.
\\
Furthermore all finite subgroups of $SU(3)$ listed in table \ref{SU3-subgroups} could be cast into one of the types listed in \cite{miller,finitesubgroups-su3-fairbairn} (see table \ref{types-SU3-subgroups}).
\\
In order not to rely on GAP and the SmallGroups Library only we developed a program (in the programming language C) which performs the following tasks:
	\begin{enumerate}
	 \item Given the generators (as $3\times 3$-matrices) of a finite group $G$ it numerically\footnote{The reason why we decided to perform a numerical analysis was of course calculation time. For some of the larger groups more than 500000 matrix multiplications were needed to obtain all group elements.} constructs all group elements in the defining representation $D$. An example for an algorithm for this purpose can be found in \cite{DTPOL}. The program uses the data type ``double'' for the real and imaginary parts of the matrix elements, respectively. An important subroutine of the program is to determine whether two matrices are equal. We decided to use the following criterion: Two matrices $A$ and $B$ are to be regarded as equal by the program if
	 \begin{equation}
	 |\mathrm{Re}(A_{ij}-B_{ij})|<10^{-7}\mbox{ and }|\mathrm{Im}(A_{ij}-B_{ij})|<10^{-7}\quad\forall i,j\in\{1,2,3\}.
	 \end{equation}
	 Using the program the orders of all groups listed in tables \ref{SU3-subgroups} and \ref{U3-subgroups} were verified (more precise: not falsified) numerically. In addition the orders of these groups were checked analytically using GAP.
	\item After the explicit construction of the defining representation $D$ of the group its character $\chi_D$ can be calculated numerically. A scalar product of the characters $\chi_D$ and $\chi_{D'}$ of two representations $D,D'$ of $G$ can be defined as
	\begin{equation}\label{SP-characters}
	(\chi_D,\chi_{D'})_G=\frac{1}{\mathrm{ord}(G)}\sum_{b\in G}\chi_{D}(b)^{\ast}\chi_{D'}(b)
	\end{equation}
	$D$ is irreducible if and only if $(\chi_D,\chi_D)_G=1$ \cite{Hamermesh}, which can easily be tested numerically. Again we regard the representation $D$ as irreducible if
		\begin{equation}
		|\mathrm{Re}(\chi_D,\chi_D)-1|<10^{-7}\mbox{ and }|\mathrm{Im}(\chi_D,\chi_D)|<10^{-7}.
		\end{equation}
	In this sense the irreducibility of all defining representations of the groups listed in tables \ref{SU3-subgroups} and \ref{U3-subgroups} was verified (more precise: not falsified) numerically. 
	\end{enumerate}
Please note that the numerical analysis described above can of course not prove the correctness of the results listed in tables \ref{SU3-subgroups} and \ref{U3-subgroups}. Note furthermore that we do not, in any sense, claim that the lists \ref{SU3-subgroups} and \ref{U3-subgroups} are complete.

\section{Construction of some series of finite subgroups of $U(3)$.}

The following theorem will allow us to construct some new infinite series of finite subgroups of $U(3)$ that have a faithful 3-dimensional irreducible representation and can not be written as a direct product with a cyclic group.

\begin{theorem}\label{GroupSeriesTheorem}
Let $G=H\rtimes\zed_n$ be a finite group with the following properties\footnote{We use the following notation for the semidirect product of two groups $A$ and $B$: $G=A\rtimes B\Rightarrow$ $A$ is a normal subgroup of $G$ and there exists a homomorphism $\phi:B\rightarrow\mathrm{Aut}(A)$. The product is defined by $(a,b)(a',b')=(a\phi(b)a',bb')$.}
	\begin{enumerate}
	 \item $G$ has a faithful $m$-dimensional irreducible representation $D$.
	 \item $n$ is prime.
	 \item The center of $G$ is of order $c\neq n$ with $c$ prime or $c=1$.
	 \item $G$ can not be written as a direct product with a cyclic group.
	\end{enumerate}
Let furthermore $A_1,...,A_a$ be generators of $D(H)$ and let $B$ be a generator of $D(\zed_n)$. Then the group\footnote{The symbol $\lang...\rang$ means ``generated by''.}
	\begin{equation}
	G_b:=\lang A_1,...,A_a,\beta B\rang,\quad\beta=e^{2\pi i/b},\quad b\in\mathbbm{N}\backslash\{0\}
	\end{equation}
(which by construction has a faithful $m$-dimensional irreducible representation) can not be written as a direct product with a cyclic group if and only if
	\begin{equation}
	b=c^j n^k,\quad j,k\in \mathbbm{N}.
	\end{equation}
\end{theorem}

\begin{theorem}\label{GroupCenterTheorem}
Let $G=H\rtimes \zed_n$ be a finite group fulfilling the properties 1.-4. of theorem \ref{GroupSeriesTheorem}. Then the center of
	\begin{equation}
	\lang A_1,...,A_a,\beta B\rang,\quad \beta=e^{2\pi i/(c^jn^k)},\quad j,k\in\mathbbm{N}
	\end{equation}
is given by
	\begin{itemize}
	 \item $\lang e^{2\pi i/c}\one_m\rang\cong\zed_c$ for $j=0,\enspace k=0$,
	 \item $\lang e^{2\pi i/c}\one_m,e^{2\pi i/n^{k-1}}\one_m\rang\cong\zed_{cn^{k-1}}$ for $j=0,\enspace k>0$,
	 \item $\lang e^{2\pi i/c^j}\one_m\rang\cong\zed_{c^j}$ for $j>0,\enspace k=0$ and
	 \item $\lang e^{2\pi i/c^j}\one_m, e^{2\pi i/n^{k-1}}\one_m\rang\cong\zed_{c^j n^{k-1}}$ \quad for $j>0,\enspace k>0$.
	\end{itemize}
If $c=1,\enspace k>0$: $G_{n^k}$ is isomorphic to $H\rtimes\zed_{n^k}$ and the center of $G_{n^k}$ is isomorphic to $\zed_{n^{k-1}}$.
\end{theorem}
\hspace{0mm}
\\
The proofs of these theorems can be found in appendix \ref{appendixC}. Let us now use theorem \ref{GroupSeriesTheorem} to construct some infinite series of finite subgroups of $U(3)$.

\subsection{The group series $T_n(m)$}

The groups $T_n$ \cite{BLW1,BLW2,Fairbairn2,DTPOL} have the structure
	\begin{equation}
	T_n=\zed_n\rtimes \zed_3,
	\end{equation}
where $n$ is a prime of the form $3k+1$. Since the center of $T_n$ is trivial we can apply theorem \ref{GroupSeriesTheorem} to find (for every $T_n$) an infinite series of finite subgroups of $U(3)$
	\begin{equation}
	T_{n}(m):=\lang e^{2\pi i/3^m}E,\enspace F(n,1,a)\rang\cong \zed_{n}\rtimes\zed_{3^m}, 
	\end{equation}
where $(1+a+a^2)\hspace{1mm}\mathrm{mod}\hspace{1mm}n=0,\enspace m\in\mathbbm{N}\backslash\{0\}.$ Theorem \ref{GroupCenterTheorem} tells us that the center of $T_n(m)$ is isomorphic to $\zed_{3^{m-1}}$.

\subsection{The group series $\Delta(3n^2,m)$}

The group $\Delta(3n^2)$ has the structure \cite{BLW1,D3n^2}
	\begin{equation}
	\Delta(3n^2)\cong (\zed_n\times \zed_n)\rtimes \zed_3, \quad n\in\mathbbm{N}\backslash\{0,1\}
	\end{equation}
and it has trivial center if $\mathrm{gcd}(n,3)=1$ \enspace\cite{D3n^2}. In the other cases one finds $c=3$, and theorem \ref{GroupSeriesTheorem} can not be applied.
\medskip
\\
Thus we find the following series of finite subgroups of $U(3)$:
	\begin{equation}
	\Delta(3n^2,m):=\lang e^{2\pi i/3^m}E,\enspace F(n,0,1)\rang\cong(\zed_n\times \zed_n)\rtimes\zed_{3^m},
	\end{equation}
where $n\in\{k\in\mathbbm{N}\vert \enspace\mathrm{gcd}(3,k)=1,k>1\},\enspace m\in\mathbbm{N}\backslash\{0\}$.

\subsection{The group series $S_{4}(m)$}\label{groupseriesS4m}
The group $S_4$ is a semidirect product of $A_4$ (which is generated by the even permutations $(14)(23)$ and $(123)$) and $\zed_2$ (generated by the odd permutation $(23)$). Since $S_4$ possesses a faithful three-dimensional irreducible representation and its center is trivial theorem \ref{GroupSeriesTheorem} leads to the following series of finite subgroups of $U(3)$:
	\begin{displaymath}
	A_4\rtimes\zed_{2^m},\quad m\in \mathbbm{N}\backslash\{0\}.
	\end{displaymath}
A faithful three-dimensional irreducible representation of $S_4$ can be obtained by reduction of the four-dimensional representation
	\begin{displaymath}
	D(\sigma)(x_1,...,x_4)^T=(x_{\sigma(1)},...,x_{\sigma(4)})^T\quad\quad \sigma\in S_4,
	\end{displaymath}
which leads to
	\begin{displaymath}
	\begin{split}
	& (14)(23)\mapsto \mbox{\begin{small}$\left(
	\begin{matrix}
	   1 & 0 & 0 \\
	   0 & -1 & 0 \\
	   0 & 0 & -1
	\end{matrix}
	\right)$\end{small}}=:A,\quad
	(123)\mapsto \mbox{\begin{small}$\left(
	\begin{matrix}
	   0 & 0 & -1 \\
	   -1 & 0 & 0 \\
	   0 & 1 & 0
	\end{matrix}
	\right)$\end{small}}=:B,\\
	& (23)\mapsto \mbox{\begin{small}$\left(
	\begin{matrix}
	   1 & 0 & 0 \\
	   0 & 0 & 1 \\
	   0 & 1 & 0
	\end{matrix}
	\right)$\end{small}}=:C.
	\end{split}
	\end{displaymath}
Thus we have found the following series of finite subgroups of $U(3)$:
	\begin{equation}
	S_{4}(m):=\lang A,B,e^{2\pi i/2^m}C\rang\cong A_4\rtimes\zed_{2^m},\quad m\in \mathbbm{N}\backslash\{0\}. 
	\end{equation}

\subsection{The group series $\Delta(6n^2,m)$ and $\Delta'(6n^2,j,k)$}

The group $\Delta(6n^2)$ has the structure \cite{BLW1,D6n^2}
	\begin{equation}
	(\zed_n\times\zed_n)\rtimes S_3
	\end{equation}
and its presentation is given by \cite{D6n^2}
	\begin{displaymath}
	\begin{split}
	& a^3=b^2=(ab)^2=\one\quad\mbox{ (presentation of $S_3$),}\\
	& c^n=d^n=\one,\quad cd=dc\quad\mbox{ (presentation of $\zed_n\times\zed_n$),}\\
	& aca^{-1}=c^{-1}d^{-1},\enspace bcb^{-1}=d^{-1},\enspace ada^{-1}=c,\enspace bdb^{-1}=c^{-1}\quad\mbox{ (semidirect product)}.
	\end{split}
	\end{displaymath}
This presentation can easily be rearranged to a presentation of
	\begin{equation}
	((\zed_n\times\zed_n)\rtimes\zed_3)\rtimes\zed_2\cong\Delta(3n^2)\rtimes\zed_2
	\end{equation}
in the following way:
	\begin{displaymath}
	\begin{split}
	& c^n=d^n=\one,\enspace cd=dc\quad\mbox{ (presentation of $\zed_n\times\zed_n$)},\\
	& a^3=\one\quad\mbox{ (presentation of $\zed_3$)},\\
	& aca^{-1}=c^{-1}d^{-1},\enspace ada^{-1}=c\quad\mbox{ (semidirect product with $\zed_3$)},\\
	& b^2=\one\quad\mbox{ (presentation of $\zed_2$)},\\
	& (ab)^2=\one \Rightarrow abab=\one\Rightarrow bab=bab^{-1}=a^{-1},\\
	& bcb^{-1}=d^{-1},\enspace bdb^{-1}=c^{-1}\quad\mbox{ (semidirect product with $\zed_2$)}.
	\end{split}
	\end{displaymath}
The center of $\Delta(6n^2)\cong\Delta(3n^2)\rtimes\zed_2$ is given by the center of $\Delta(3n^2)$, which can be of order $1$ or $3$. Thus we can apply theorem \ref{GroupSeriesTheorem} to construct new series of finite subgroups of $U(3)$.
\medskip
\\
A faithful three-dimensional irreducible representation of $\Delta(6n^2)$ is given by \cite{D6n^2}
	\begin{displaymath}
	a\mapsto \mbox{\begin{small}$\left(
	\begin{matrix}
	   0 & 1 & 0 \\
	   0 & 0 & 1 \\
	   1 & 0 & 0
	\end{matrix}
	\right)$\end{small}}=:A,\enspace
	b\mapsto \mbox{\begin{small}$\left(
	\begin{matrix}
	   0 & 0 & -1 \\
	   0 & -1 & 0 \\
	   -1 & 0 & 0
	\end{matrix}
	\right)$\end{small}}=:B,\enspace
	d\mapsto \mbox{\begin{small}$\left(
	\begin{matrix}
	   1 & 0 & 0 \\
	   0 & \eta & 0 \\
	   0 & 0 & \eta^{\ast}
	\end{matrix}
	\right)$\end{small}}=:C,
	\end{displaymath}
where $\eta=e^{2\pi i/n}$ and $n\in\mathbbm{N}\backslash\{0,1\}$.
\medskip
\\
There are two possibilities:
	\begin{itemize}
	 \item $\mathrm{gcd}(3,n)=1\Rightarrow$ The center of $\Delta(6n^2)$ is trivial, which leads to the group series
		\begin{equation}
		\Delta(6n^2,m):=\lang A, e^{2\pi i/2^m}B, C\rang\cong\Delta(3n^2)\rtimes\zed_{2^m},
		\end{equation}
		$n\in\{k\in\mathbbm{N}\vert \enspace\mathrm{gcd}(3,k)=1,k>1\},\enspace m\in\mathbbm{N}\backslash\{0\}$. This series contains $S_4(m)=\Delta(6\times 2^2,m)$ as a subseries.
	 \item $\mathrm{gcd}(3,n)=3\Rightarrow$ The center of $\Delta(6n^2)$ is of order 3, which leads to the group series
		\begin{equation}
		\Delta'(6n^2,j,k):=\lang A, e^{2\pi i/(3^j 2^k)}B, C\rang,
		\end{equation}
		$n\in\{k\in\mathbbm{N}\vert \enspace\mathrm{gcd}(3,k)=3\},\enspace j,k\in\mathbbm{N}\backslash\{0\}.$
	\end{itemize}

\section{Analysis of two interesting finite subgroups of $U(3)$}

\subsection{The group $\mbox{\textlbrackdbl}27,4\mbox{\textrbrackdbl}$}

The group $\mbox{\textlbrackdbl}27,4\mbox{\textrbrackdbl}$ is the smallest group listed in table \ref{U3-subgroups}. It is generated by
	\begin{displaymath}
	A:=R(3,1,1,2)=\mbox{\begin{small}$\left(
	\begin{matrix}
	   0 & 0 & \omega \\
	   \omega & 0 & 0 \\
	   0 & \omega^2 & 0
	\end{matrix}
	\right)$\end{small}},\quad B:=R(3,1,2,1)=\mbox{\begin{small}$\left(
	\begin{matrix}
	   0 & 0 & \omega \\
	   \omega^2 & 0 & 0 \\
	   0 & \omega & 0
	\end{matrix}
	\right)$\end{small}}.
	\end{displaymath}
A much simpler set of generators is given by
	\begin{equation}\label{274generators}
	R:=BA=\mbox{\begin{small}$\left(
	\begin{matrix}
	   0 & 1 & 0 \\
	   0 & 0 & 1 \\
	   \omega^2 & 0 & 0
	\end{matrix}
	\right)$\end{small}},\quad S:=AB^{-1}=\mbox{\begin{small}$\left(
	\begin{matrix}
	   1 & 0 & 0 \\
	   0 & \omega^2 & 0 \\
	   0 & 0 & \omega
	\end{matrix}
	\right)$\end{small}},
	\end{equation}
which is a generating set, because $(SR)^5=A$ and $R(SR)^{-5}=B$. From the generators $R$ and $S$ given in equation (\ref{274generators}) we find that
	\begin{displaymath}
	\lang R\rang\cong\zed_9,\quad \lang S\rang\cong\zed_3,\quad \lang R\rang\cap\lang S\rang=\{\one_3\},\quad S^{-1}RS=R^4,
	\end{displaymath}
thus
	\begin{equation}
	\mbox{\textlbrackdbl}27,4\mbox{\textrbrackdbl}:=\lang R,S\rang\cong \zed_9\rtimes\zed_3.
	\end{equation}
Due to the semidirect product structure (especially using the fact that every element of $\mbox{\textlbrackdbl}27,4\mbox{\textrbrackdbl}$ can be written in the form $\omega^x R^y S^z$) the derivation of the conjugacy classes is straight forward. One finds eleven conjugacy classes
	\begin{equation}
	\begin{array}{l}
	C_1 =\{\one_3\},\\
	C_2 =\{\omega \one_3\},\\
	C_3 =\{\omega^2  \one_3\},\\
	C_4 =\{R,\omega R,\omega^2 R\},\\
	C_5 =\{R^2,\omega R^2,\omega^2 R^2\},\\
	C_6 =\{S,\omega S,\omega^2 S\},\\
	C_7 =\{S^2,\omega S^2,\omega^2 S^2\},\\
	C_8 =\{RS,\omega RS,\omega^2 RS\},\\
	C_9 =\{RS^2,\omega RS^2,\omega^2 RS^2 \},\\
	C_{10} =\{R^2 S,\omega R^2 S,\omega^2 R^2 S\},\\
	C_{11} =\{R^2 S^2,\omega R^2 S^2,\omega^2 R^2 S^2\}.
	\end{array}
	\end{equation}
The nontrivial normal subgroups are found to be
	\begin{equation}
	\begin{array}{l}
	C_1\cup C_2\cup C_3=\lang \omega\one_3 \rang\cong\zed_3,\\
	C_1\cup C_2\cup C_3\cup C_4\cup C_5=\lang R\rang\cong\zed_9,\\
	C_1\cup C_2\cup C_3\cup C_6\cup C_7=\lang \omega\one_m,S\rang\cong\zed_3\times\zed_3\\
	C_1\cup C_2\cup C_3\cup C_8\cup C_{11}=\lang RS,R^2 S^2\rang\cong\zed_9,\\
	C_1\cup C_2\cup C_3\cup C_9\cup C_{10}=\lang RS^2,R^2 S\rang\cong\zed_9.
	\end{array}
	\end{equation}
Since there are eleven conjugacy classes there must be eleven inequivalent irreducible representations. These are the nine one-dimensional irreducible representations of the factor group
	\begin{displaymath}
	\lang R,S\rang/\lang \omega \one_3\rang\cong \zed_3\times\zed_3,
	\end{displaymath}
the defining representation and its complex conjugate:
	\begin{subequations}
	\begin{eqnarray}
	&& \textbf{\underline{1}}_{(i,j)}:\quad R\mapsto \omega^i,\enspace S\mapsto \omega^j,\quad (i,j=0,1,2),\\
	&& \textbf{\underline{3}}:\quad R\mapsto R,\enspace S\mapsto S,\\
	&& \textbf{\underline{3}}^{\ast}:\quad R\mapsto R^{\ast},\enspace S\mapsto S^{\ast}.
	\end{eqnarray}
	\end{subequations}
The character table of $\mbox{\textlbrackdbl}27,4\mbox{\textrbrackdbl}$ is shown in table \ref{charactertable274}.

\begin{table}
\begin{center}
\begin{tabular}{||c||ccccccccccc||}
\hline \hline
$\mbox{\textlbrackdbl}27,4\mbox{\textrbrackdbl}$ & 
$C_1$ & $C_2$ & $C_3$ & $C_4$ & $C_5$ & $C_6$ & $C_7$ & $C_8$ &
$C_9$ & $C_{10}$ & $C_{11}$\\
(\# $C_k$) & (1) & (1) & (1) & (3) & (3) & (3) & (3) & (3) &
(3) & (3) & (3) \\
$\mathrm{ord}(C_k)$ & 1 & 3 & 3 & 9 & 9 & 3 & 3 & 9 & 9 & 9 & 9\\
\hline \hline
$\mathbf{1}_{(0,0)}$ &
$1$ & $1$ & $1$ &
$1$ & $1$ & $1$ &
$1$ & $1$ & $1$ &
$1$ & $1$\\
$\mathbf{1}_{(0,1)}$ &
$1$ & $1$ & $1$ &
$1$ & $1$ & $\omega$ &
$\omega^2$ & $\omega$ & $\omega^2$ &
$\omega$ & $\omega^2$\\
$\mathbf{1}_{(0,2)}$ &
$1$ & $1$ & $1$ &
$1$ & $1$ & $\omega^2$ &
$\omega$ & $\omega^2$ & $\omega$ &
$\omega^2$ & $\omega$\\
$\mathbf{1}_{(1,0)}$ &
$1$ & $1$ & $1$ &
$\omega$ & $\omega^2$ & $1$ &
$1$ & $\omega$ & $\omega$ &
$\omega^2$ & $\omega^2$\\
$\mathbf{1}_{(1,1)}$ &
$1$ & $1$ & $1$ &
$\omega$ & $\omega^2$ & $\omega$ &
$\omega^2$ & $\omega^2$ & $1$ &
$1$ & $\omega$\\
$\mathbf{1}_{(1,2)}$ &
$1$ & $1$ & $1$ &
$\omega$ & $\omega^2$ & $\omega^2$ &
$\omega$ & $1$ & $\omega^2$ &
$\omega$ & $1$\\
$\mathbf{1}_{(2,0)}$ &
$1$ & $1$ & $1$ &
$\omega^2$ & $\omega$ & $1$ &
$1$ & $\omega^2$ & $\omega^2$ &
$\omega$ & $\omega$\\
$\mathbf{1}_{(2,1)}$ &
$1$ & $1$ & $1$ &
$\omega^2$ & $\omega$ & $\omega$ &
$\omega^2$ & $1$ & $\omega$ &
$\omega^2$ & $\omega^2$\\
$\mathbf{1}_{(2,2)}$ &
$1$ & $1$ & $1$ &
$\omega^2$ & $\omega$ & $\omega^2$ &
$\omega$ & $\omega$ & $1$ &
$1$ & $1$\\
\hline
$\textbf{\underline{3}}$ & $3$ & $3\omega$ & $3\omega^2$ & 0 & 0 & 0 & 0 & 0 & 0 & 0 & 0\\
$\textbf{\underline{3}}^{\ast}$ & $3$ & $3\omega^2$ & $3\omega$ & 0 & 0 & 0 & 0 & 0 & 0 & 0 & 0\\
\hline \hline
\end{tabular}
\caption{Character table of $\mbox{\textlbrackdbl}27,4\mbox{\textrbrackdbl}$. 
The number of elements in each class is given by the numbers in
parentheses in the second line of the table.
\label{charactertable274}}
\end{center}
\end{table}
\hspace{0mm}\\
The tensor products are given by
	\begin{subequations}
	\begin{eqnarray}
	&& \textbf{\underline{1}}_{(i,j)}\otimes\textbf{\underline{1}}_{(k,l)}=\textbf{\underline{1}}_{((i+k)\hspace{0.5mm}\mathrm{mod}3,\hspace{0.5mm}(j+l)\hspace{0.5mm}\mathrm{mod}3)},\quad (i,j,k,l=0,1,2)\\
	&& \textbf{\underline{3}}\otimes\textbf{\underline{3}}=\textbf{\underline{3}}^{\ast}\oplus\textbf{\underline{3}}^{\ast}\oplus\textbf{\underline{3}}^{\ast},\\
	&& \textbf{\underline{3}}^{\ast}\otimes\textbf{\underline{3}}^{\ast}=\textbf{\underline{3}}\oplus\textbf{\underline{3}}\oplus\textbf{\underline{3}},\\
	&& \textbf{\underline{3}}\otimes\textbf{\underline{3}}^{\ast}=\bigoplus_{i,j=0}^2 \textbf{\underline{1}}_{(i,j)}.
	\end{eqnarray}
	\end{subequations}
The corresponding invariant subspaces are given by
	\begin{subequations}
	\begin{eqnarray}
	&& V_{\textbf{\underline{3}}\otimes\textbf{\underline{3}}\rightarrow \textbf{\underline{3}}^{\ast}}=\mathrm{Span}\left(e_1\otimes e_1,\enspace e_2\otimes e_2,\enspace e_3\otimes e_3\right),\\
	&& V_{\textbf{\underline{3}}\otimes\textbf{\underline{3}}\rightarrow \textbf{\underline{3}}^{\ast}}'=\mathrm{Span}\left(e_2\otimes e_3,\enspace \omega e_3\otimes e_1,\enspace \omega^2 e_1\otimes e_2\right),\\
	&& V_{\textbf{\underline{3}}\otimes\textbf{\underline{3}}\rightarrow \textbf{\underline{3}}^{\ast}}''=\mathrm{Span}\left(e_3\otimes e_2,\enspace \omega e_1\otimes e_3,\enspace \omega^2 e_2\otimes e_1\right),\\
	&& V_{\textbf{\underline{3}}\otimes \textbf{\underline{3}}^{\ast}\rightarrow \textbf{\underline{1}}_{(0,0)}}=\mathrm{Span}\left(e_1\otimes e_1+e_2\otimes e_2+e_3\otimes e_3\right),\\
	&& V_{\textbf{\underline{3}}\otimes \textbf{\underline{3}}^{\ast}\rightarrow \textbf{\underline{1}}_{(0,1)}}=\mathrm{Span}\left(e_1\otimes e_2+e_2\otimes e_3+\omega^2 e_3\otimes e_1\right),\\
	&& V_{\textbf{\underline{3}}\otimes \textbf{\underline{3}}^{\ast}\rightarrow \textbf{\underline{1}}_{(0,2)}}=\mathrm{Span}\left(e_1\otimes e_3+\omega^2 e_2\otimes e_1+\omega^2 e_3\otimes e_2\right),\\
	&& V_{\textbf{\underline{3}}\otimes \textbf{\underline{3}}^{\ast}\rightarrow \textbf{\underline{1}}_{(1,0)}}=\mathrm{Span}\left(e_1\otimes e_1+\omega e_2\otimes e_2+\omega^2 e_3\otimes e_3\right),\\
	&& V_{\textbf{\underline{3}}\otimes \textbf{\underline{3}}^{\ast}\rightarrow \textbf{\underline{1}}_{(1,1)}}=\mathrm{Span}\left(e_1\otimes e_2+\omega e_2\otimes e_3+\omega e_3\otimes e_1\right),\\
	&& V_{\textbf{\underline{3}}\otimes \textbf{\underline{3}}^{\ast}\rightarrow \textbf{\underline{1}}_{(1,2)}}=\mathrm{Span}\left(e_1\otimes e_3+e_2\otimes e_1+\omega e_3\otimes e_2\right),\\
	&& V_{\textbf{\underline{3}}\otimes \textbf{\underline{3}}^{\ast}\rightarrow \textbf{\underline{1}}_{(2,0)}}=\mathrm{Span}\left(e_1\otimes e_1+\omega^2 e_2\otimes e_2+\omega e_3\otimes e_3\right),\\
	&& V_{\textbf{\underline{3}}\otimes \textbf{\underline{3}}^{\ast}\rightarrow \textbf{\underline{1}}_{(2,1)}}=\mathrm{Span}\left(e_1\otimes e_2+\omega^2 e_2\otimes e_3+e_3\otimes e_1\right),\\
	&& V_{\textbf{\underline{3}}\otimes \textbf{\underline{3}}^{\ast}\rightarrow \textbf{\underline{1}}_{(2,2)}}=\mathrm{Span}\left(e_1\otimes e_3+\omega e_2\otimes e_1+e_3\otimes e_2\right).
	\end{eqnarray}
	\end{subequations}
Since the defining representations of $\mbox{\textlbrackdbl}27,4\mbox{\textrbrackdbl}$ and $\Delta(27)$ differ by phase factors only,
	 \begin{equation}
	 \lang e^{2\pi i/9}R,S\rang\cong\Delta(27),
	 \end{equation}
all Clebsch-Gordan coefficients for corresponding tensor product decompositions are equal.
\\
The structure of $\mbox{\textlbrackdbl}27,4\mbox{\textrbrackdbl}\cong\zed_9\rtimes\zed_3$ is very similar to the structure of $\Delta(27)\cong (\zed_3\times\zed_3)\rtimes\zed_3$. Though these groups are not isomorphic they share the nine one-dimensional irreducible representations as well as the character table (except for the values of $\mathrm{ord}(C_k)$). The character table of $\Delta(27)$ can be found in \cite{Principal_series,Delta(27)b}. Since the character tables of $\mbox{\textlbrackdbl}27,4\mbox{\textrbrackdbl}$ and $\Delta(27)$ are equal, all tensor products are equal. Since also the Clebsch-Gordan coefficients are equal we find that, from the point of view of model building, $\Delta(27)$ and $\mbox{\textlbrackdbl}27,4\mbox{\textrbrackdbl}$ are equivalent.

\subsection{The group $S_4(2)\cong A_4\rtimes\zed_4$}

The group $S_4\cong A_4\rtimes\zed_2$ has been commonly used in model building, and especially in the last years interest in $S_4$ began to increase \cite{S4,s4-papers,s4papers-2,s4-papers-3}. Therefore the group $S_4(2)$, which is a relative of $S_4$, may be of interest. From subsection \ref{groupseriesS4m} we know the structure
	\begin{equation}
	S_4(2)\cong A_4\rtimes\zed_4 
	\end{equation}
and generators of a faithful three-dimensional irreducible representation of $S_4(2)$:
	\begin{equation}
	 A:=\mbox{\begin{small}$\left(
	\begin{matrix}
	   1 & 0 & 0 \\
	   0 & -1 & 0 \\
	   0 & 0 & -1
	\end{matrix}
	\right)$\end{small}},\quad
	B:=\mbox{\begin{small}$\left(
	\begin{matrix}
	   0 & 0 & -1 \\
	   -1 & 0 & 0 \\
	   0 & 1 & 0
	\end{matrix}
	\right)$\end{small}},\quad
	C:=\mbox{\begin{small}$i\left(
	\begin{matrix}
	   1 & 0 & 0 \\
	   0 & 0 & 1 \\
	   0 & 1 & 0
	\end{matrix}
	\right)$\end{small}}.
	\end{equation}
$S_4$ possesses five conjugacy classes $C_i$. The classes of $S_4(2)$ differ from $C_i$ by the element $C^2=-\one_3$ only, thus one finds ten conjugacy classes of $S_4(2)$:
	\begin{equation}
	\begin{array}{l}
	C^1_{\pm} =\{\pm\one_3\},\\
	C^2_{\pm} =\{\pm A,\pm BAB^2,\pm B^2 AB\},\\
	C^3_{\pm} =\{\pm B,\pm AB,\pm BA,\pm ABA,\pm B^2,\pm BAB,\pm B^2 A,\pm AB^2\},\\
	C^4_{\pm} =\{\pm C,\pm B^2 C,\pm BABC,\pm BC,\pm AC,\pm ABAC\},\\
	C^5_{\pm} =\{\pm AB^2 C,\pm B^2 AC,\pm B^2 ABC,\pm ABC,\pm BAB^2 C,\pm BAC\}.
	\end{array}
	\end{equation}
The nontrivial normal subgroups of $S_4(2)$ are given by
	\begin{equation}
	\begin{array}{l}
	C^1_+\cup C^1_-\cong \zed_2,\\
	C^1_+\cup C^2_+\cong \zed_2\times\zed_2,\\
	C^1_+\cup C^1_-\cup C^2_+\cup C^2_-\cong\zed_2\times\zed_2\times\zed_2,\\
	C^1_+\cup C^2_+\cup C^3_+\cong A_4,\\
	C^1_+\cup C^1_-\cup C^2_+\cup C^2_-\cup C^3_+\cup C^3_-\cong A_4\times\zed_2.
	\end{array}
	\end{equation}
Since $S_4(2)/\{\one_3,-\one_3\}\cong(A_4\rtimes\zed_4)/\zed_2\cong A_4\rtimes\zed_2\cong S_4$ we find that all irreducible representations of $S_4$ are irreducible representations of $S_4(2)$ too. By construction $A_4$ is an invariant subgroup of $S_4(2)$, thus all irreducible representations of $\zed_4\cong S_4(2)/A_4$ are irreducible representations of $S_4(2)$ too. They are given by:
	  \begin{subequations}
	  \begin{eqnarray}
	  && \textbf{\underline{1}}_1:\enspace A\mapsto 1,\enspace B\mapsto 1,\enspace C\mapsto 1,\\
	  && \textbf{\underline{1}}_2:\enspace A\mapsto 1,\enspace B\mapsto 1,\enspace C\mapsto -1,\\
	  && \textbf{\underline{1}}_3:\enspace A\mapsto 1,\enspace B\mapsto 1,\enspace C\mapsto i,\\
	  && \textbf{\underline{1}}_4:\enspace A\mapsto 1,\enspace B\mapsto 1,\enspace C\mapsto -i.
	  \end{eqnarray}
	  \end{subequations}
We will now construct the irreducible representations of $S_4$: Since $S_4/A_4\cong\zed_2$ we obtain two one-dimensional irreducible representations
	  \begin{subequations}
	  \begin{eqnarray}
	  && \textbf{\underline{1}}_1:\enspace A\mapsto 1,\enspace B\mapsto 1,\enspace C\mapsto 1,\\
	  && \textbf{\underline{1}}_2:\enspace A\mapsto 1,\enspace B\mapsto 1,\enspace C\mapsto -1,
	  \end{eqnarray}
	  \end{subequations}
which are irreducible representations of $\zed_4$ also. Multiplying these one-dimensional representations with the defining representation of $S_4$ we obtain the two three-dimensional irreducible representations of $S_4$:
	 \begin{subequations}
	  \begin{eqnarray}
	  && \textbf{\underline{3}}_1:\enspace A\mapsto A,\enspace B\mapsto B,\enspace C\mapsto -iC,\\
	  && \textbf{\underline{3}}_2:\enspace A\mapsto A,\enspace B\mapsto B,\enspace C\mapsto iC.
	  \end{eqnarray}
	  \end{subequations}
The missing two-dimensional irreducible representation can be obtained in the following way: The Klein four-group $\zed_2\times\zed_2$ is an invariant subgroup of $S_4$:
	\begin{equation}\label{Klein-fourgroup}
	C^1_+\cup C^2_{+}=\{\one_3, A, BAB^2, B^2 AB\}\cong \zed_2\times \zed_2.
	\end{equation}
Therefore all irreducible representations of $S_3\cong S_4/(\zed_2\times\zed_2)$ are irreducible representations of $S_4$ too. This has also been pointed out in \cite{s4papers-2}. Assuming that $\zed_2\times\zed_2$ given in equation (\ref{Klein-fourgroup}) is mapped onto $\one_3$ one can easily construct the three-dimensional \textit{reducible} $S_3$-representation
	\begin{equation}
	\textbf{\underline{3}}_r:\enspace A\mapsto \one_3,\enspace B\mapsto B,\enspace C\mapsto -iC.
	\end{equation}
$v=\frac{1}{\sqrt{3}}(1,-1,-1)^T$ is a common eigenvector of $B$ and $-iC$ to the eigenvalue 1. This enables reduction via
	\begin{displaymath}
	U:=\mbox{\begin{small}$\left(
	\begin{matrix}
	   \frac{1}{\sqrt{3}} & \sqrt{\frac{2}{3}} & 0 \\
	   -\frac{1}{\sqrt{3}} & \frac{1}{\sqrt{6}} & \frac{1}{\sqrt{2}} \\
	   -\frac{1}{\sqrt{3}} & \frac{1}{\sqrt{6}} & -\frac{1}{\sqrt{2}}
	\end{matrix}
	\right)$\end{small}}.
	\end{displaymath}
	\begin{displaymath}
	U^T BU=\mbox{\begin{small}$\left(
	\begin{matrix}
	   1 & 0 & 0 \\
	   0 & -\frac{1}{2} & \frac{\sqrt{3}}{2} \\
	   0 & -\frac{\sqrt{3}}{2} & -\frac{1}{2}
	\end{matrix}
	\right)$
	\end{small}},\quad U^T (-iC)U=\mbox{\begin{small}$\left(
	\begin{matrix}
	   1 & 0 & 0 \\
	   0 & 1 & 0 \\
	   0 & 0 & -1
	\end{matrix}
	\right)$\end{small}}.
	\end{displaymath}
Thus the two-dimensional irreducible representation of $S_4$ is given by
	\begin{equation}
	\textbf{\underline{2}}:\enspace A\mapsto \one_2,\enspace B\mapsto \mbox{\begin{small}$\left(
	\begin{matrix}
	   -\frac{1}{2} & \frac{\sqrt{3}}{2} \\
	   -\frac{\sqrt{3}}{2} & -\frac{1}{2}
	\end{matrix}
	\right)$
	\end{small}},\enspace -iC\mapsto\mbox{\begin{small}$\left(
	\begin{matrix}
	1 & 0\\
	0 & -1
	\end{matrix}
	\right)$
	\end{small}}.
	\end{equation} 
The irreducible representations of $S_4(2)$ are thus given by
	\begin{subequations}
	\begin{eqnarray}
	&& \textbf{\underline{1}}_1:\enspace A\mapsto 1,\enspace B\mapsto 1,\enspace C\mapsto 1,\\
	  && \textbf{\underline{1}}_2:\enspace A\mapsto 1,\enspace B\mapsto 1,\enspace C\mapsto -1,\\
	  && \textbf{\underline{1}}_3:\enspace A\mapsto 1,\enspace B\mapsto 1,\enspace C\mapsto i,\\
	  && \textbf{\underline{1}}_4:\enspace A\mapsto 1,\enspace B\mapsto 1,\enspace C\mapsto -i,\\
	&& \textbf{\underline{2}}_1:\enspace A\mapsto \one_2,\enspace B\mapsto \mbox{\begin{small}$\left(
	\begin{matrix}
	   -\frac{1}{2} & \frac{\sqrt{3}}{2} \\
	   -\frac{\sqrt{3}}{2} & -\frac{1}{2}
	\end{matrix}
	\right)$
	\end{small}},\enspace C\mapsto\mbox{\begin{small}$\left(
	\begin{matrix}
	1 & 0\\
	0 & -1
	\end{matrix}
	\right)$
	\end{small}},\\
	&& \textbf{\underline{2}}_2:\enspace A\mapsto \one_2,\enspace B\mapsto \mbox{\begin{small}$\left(
	\begin{matrix}
	   -\frac{1}{2} & \frac{\sqrt{3}}{2} \\
	   -\frac{\sqrt{3}}{2} & -\frac{1}{2}
	\end{matrix}
	\right)$
	\end{small}},\enspace C\mapsto\mbox{\begin{small}$i\left(
	\begin{matrix}
	1 & 0\\
	0 & -1
	\end{matrix}
	\right)$
	\end{small}},\\
	&& \textbf{\underline{3}}_1:\enspace A\mapsto A,\enspace B\mapsto B,\enspace C\mapsto -iC,\\
	  && \textbf{\underline{3}}_2:\enspace A\mapsto A,\enspace B\mapsto B,\enspace C\mapsto iC,\\
	&& \textbf{\underline{3}}_3:\enspace A\mapsto A,\enspace B\mapsto B,\enspace C\mapsto C,\\
	  && \textbf{\underline{3}}_4:\enspace A\mapsto A,\enspace B\mapsto B,\enspace C\mapsto -C.
	\end{eqnarray}
	\end{subequations}
The extension of this analysis to $S_4(m)$ is easy - one just needs to take the irreducible representations of $S_4(m)/A_4\cong\zed_{2^{m}}$ and multiply them with the irreducible representations of $S_4$ to obtain all irreducible representations of $S_4(m)$.
\medskip
\\
From this it is clear that also all tensor product decompositions and corresponding invariant subspaces have the same structure as those of $S_4$. All $3\otimes 3$-tensor products can be constructed from the $3\otimes 3$-tensor product
	\begin{equation}\label{CGCS4}
	\textbf{\underline{3}}_1\otimes \textbf{\underline{3}}_1=\textbf{\underline{1}}_1\oplus \textbf{\underline{2}}_1\oplus \textbf{\underline{3}}_1\oplus \textbf{\underline{3}}_2
	\end{equation}
of $S_4$ by multiplication with one-dimensional irreducible representations of $S_4(2)$. The corresponding invariant subspaces for the Clebsch-Gordan decomposition (\ref{CGCS4}) are spanned by the following vectors \cite{DTPOL}:
	\begin{subequations}
	\begin{eqnarray}
	&& v(\textbf{\underline{1}}_1)=\frac{1}{\sqrt{3}}\left( e_1\otimes e_1+ e_2\otimes e_2+e_3\otimes e_3\right), \\
	&& v_1(\textbf{\underline{2}}_1)=\frac{1}{\sqrt{6}}\left( -2 e_1\otimes e_1+e_2\otimes e_2+e_3\otimes e_3\right),\\
	&& v_2(\textbf{\underline{2}}_1)=\frac{1}{\sqrt{2}}\left(e_2\otimes e_2-e_3\otimes e_3 \right),\\
	&& v_1(\textbf{\underline{3}}_1)=\frac{1}{\sqrt{2}}\left(e_2\otimes e_3+e_3\otimes e_2 \right),\\
	&& v_2(\textbf{\underline{3}}_1)=\frac{1}{\sqrt{2}}\left(e_1\otimes e_3+e_3\otimes e_1 \right),\\
	&& v_3(\textbf{\underline{3}}_1)=\frac{1}{\sqrt{2}}\left(e_1\otimes e_2+e_2\otimes e_1 \right),\\
	&& v_1(\textbf{\underline{3}}_2)=\frac{1}{\sqrt{2}}\left(e_2\otimes e_3-e_3\otimes e_2 \right),\\
	&& v_2(\textbf{\underline{3}}_2)=\frac{1}{\sqrt{2}}\left(-e_1\otimes e_3+e_3\otimes e_1 \right),\\
	&& v_3(\textbf{\underline{3}}_2)=\frac{1}{\sqrt{2}}\left(e_1\otimes e_2-e_2\otimes e_1 \right).
	\end{eqnarray}
	\end{subequations}
Let us finally investigate the differences between the symmetry groups $S_4$ and $S_4(2)$ from the physical point of view.
\medskip
\\
Let us as an example consider a field theory describing seven real scalar fields arranged in the following $S_4(2)$-multiplets:
	\begin{equation}
	\textbf{\underline{3}}_1: \phi=\left(\begin{matrix} \phi_1 \\ \phi_2 \\ \phi_3 \end{matrix}\right),\quad \textbf{\underline{3}}_3: \psi=\left(\begin{matrix} \psi_1 \\ \psi_2 \\ \psi_3 \end{matrix}\right),\quad \textbf{\underline{1}}_3: \eta.
	\end{equation}
The Lagrangian
	\begin{equation}\label{Lagrangian1}
	\mathcal{L}_1=\phi^T \psi\eta + \eta^4
	\end{equation}
is invariant under this transformation, while
	\begin{equation}\label{Lagrangian2}
	\mathcal{L}_2=\eta^2
	\end{equation}
clearly is not. $\mathcal{L}_2$ can not be forbidden in a pure $S_4$-theory (allowing $\mathcal{L}_1$), because for this issue one needs the $\zed_4$-representation $\textbf{\underline{1}}_3$ of $S_4(2)$, which is not contained in $S_4$. Another group based on $S_4$ containing the needed $\zed_4$-representation is $S_4\times\zed_4$. The 20 irreducible representations of $S_4\times\zed_4$ are given by ($j=0,1,2,3$):
	\begin{subequations}\label{S4xZ4irreps}
	\begin{eqnarray}
	&& \textbf{\underline{1}}_{1j}:\enspace a\mapsto 1,\enspace b\mapsto 1,\enspace c\mapsto 1,\enspace d\mapsto i^j,\\
	  && \textbf{\underline{1}}_{2j}:\enspace a\mapsto 1,\enspace b\mapsto 1,\enspace c\mapsto -1,\enspace d\mapsto i^j,\\
	&& \textbf{\underline{2}}_{j}:\enspace a\mapsto \one_2,\enspace b\mapsto \mbox{\begin{small}$\left(
	\begin{matrix}
	   -\frac{1}{2} & \frac{\sqrt{3}}{2} \\
	   -\frac{\sqrt{3}}{2} & -\frac{1}{2}
	\end{matrix}
	\right)$
	\end{small}},\enspace c\mapsto\mbox{\begin{small}$\left(
	\begin{matrix}
	1 & 0\\
	0 & -1
	\end{matrix}
	\right)$
	\end{small}}, \enspace d\mapsto i^j\one_2,\\
	&& \textbf{\underline{3}}_{1j}:\enspace a\mapsto A,\enspace b\mapsto B,\enspace c\mapsto -iC,\enspace d\mapsto i^j\one_3\\
	  && \textbf{\underline{3}}_{2j}:\enspace a\mapsto A,\enspace b\mapsto B,\enspace c\mapsto iC,\enspace d\mapsto i^j\one_3.
	\end{eqnarray}
	\end{subequations}
From equations (\ref{S4xZ4irreps}) it is clear that all irreducible representations of $S_4\times\zed_4$ can be interpreted as irreducible representations of $S_4(2)$ with an additional generator $d$, which acts as multiplication with a phase factor. Indeed $S_4(2)$ is a subgroup of $S_4\times\zed_4$. It is thus clear that all Lagrangians based on a symmetry group $S_4\times\zed_4$ are allowed in the corresponding $S_4(2)$-theory too. The question remains whether there are $S_4(2)$-models which do not fit to an appropriate $S_4\times\zed_4$-model. The answer is no for the following reason: Consider a Lagrangian
	\begin{equation}\label{Lagrangian}
	\mathcal{L}(\phi_1,...,\phi_m)=\sum_{j} \mathcal{L}_j(\phi_1,...,\phi_m)
	\end{equation}
invariant under the action of a symmetry group $G$:
	\begin{equation}\label{Group_action}
	G\ni a:\enspace \phi_i\mapsto D_i(a)\phi_i\quad i=1,...,m,
	\end{equation}
where $D_j$ are representations of $G$. The Lagrangians $\mathcal{L}_j$ are assumed to fulfill the following properties:
	\begin{itemize}
	 \item $\mathcal{L}_j$ is invariant under the action (\ref{Group_action}) of $G$.
	 \item $\mathcal{L}_j$ can not be split up into two ``smaller'' Lagrangians being invariant under $G$ themselves\footnote{E.g. $\mathcal{L}(\eta)=\eta^2+\eta^4$ is invariant under $\zed_2:\enspace \eta\mapsto -\eta$, but it can be split up into the two ``smaller'' Lagrangians $\mathcal{L}_1(\eta)=\eta^2$ and $\mathcal{L}_2(\eta)=\eta^4$. $\nexists k\in\mathbbm{N}$ such that $\mathcal{L}(\alpha\eta)=\alpha^k\mathcal{L}(\eta) \enspace\forall\alpha\in U(1)$, while $\mathcal{L}_1(\alpha\eta)=\alpha^2\mathcal{L}_1(\eta)$ and $\mathcal{L}_2(\alpha\eta)=\alpha^4\mathcal{L}_2(\eta)\enspace\forall\alpha\in U(1)$.}. More precise: $\exists k\in\mathbbm{N}$ such that $\forall \alpha\in U(1):$ $\mathcal{L}_j(\alpha\phi_1,....,\alpha\phi_m)=\alpha^k\mathcal{L}_j(\phi_1,...,\phi_m)$.
	\end{itemize}
The construction of an invariant Lagrangian (\ref{Lagrangian}) can then be split up into two steps:
 	\begin{enumerate}
 	 \item $\mathcal{L}_j$ must transform as one-dimensional representations of $G$.
 	 \item If possible, the chosen representations $D_1,...,D_m$ have to be multiplied by one-dimensional representations in such a way that $\mathcal{L}_j$ are invariant under $G$. If this is not possible $\mathcal{L}_j$ is forbidden by the symmetry $G$.
 	\end{enumerate}
In this language the problem of the relation between $S_4(2)$ and $S_4\times\zed_4$ can be reformulated as follows: Suppose a Lagrangian $\mathcal{L}$ invariant under the action of $S_4(2)$ is given. Since the irreducible representations of $S_4(2)$ and $S_4\times\zed_4$ differ by phase factors only, we find that point 1. stated above is fulfilled automatically. We can now replace any irreducible representation $D_i$ of $S_4(2)$ containing elements of the form ``real matrix times $\pm i$'' by the corresponding irreducible representation of $S_4\times\zed_4$ containing all four elements of the center. In this case all Lagrangians $\mathcal{L}_j$ will remain invariant, because in order to construct $\mathcal{L}_j$ invariant under $S_4(2)$ one already had to take care of the phase factor $i$ contained in the element $C$ of $S_4(2)$. Thus from the point of view of invariant Lagrangians (which is the interesting point of view for physics), $S_4\times\zed_4$ and $S_4(2)$ are equivalent.
\bigskip
\\
In this section we have analysed two interesting finite subgroups of $U(3)$. It turned out that in the case of these two groups it is possible to find (in the case of $\mbox{\textlbrackdbl}27,4\mbox{\textrbrackdbl}$) a finite subgroup of $SU(3)$ or (in the case of $S_4(2)$) a direct product of a finite subgroup of $SU(3)$ with a cyclic group which is equivalent to the $U(3)$-subgroup from the physical point of view, i.e. which allows the same Lagrangians. The question remains whether this is true in general.

\section{Conclusions}
In this work we used the SmallGroups Library \cite{SmallGroups,SmallGroupsLibrary} to find the finite subgroups of $U(3)$ of order smaller than 512. Using the computer algebra system GAP it was possible to construct generators for all these groups.
\medskip
\\
Inspired by the results (see tables  \ref{SU3-subgroups} and \ref{U3-subgroups}) of this analysis we developed the two theorems \ref{GroupSeriesTheorem} and \ref{GroupCenterTheorem} which led to the discovery of the series of finite subgroups of $U(3)$
	\begin{displaymath}
	T_n(m),\quad \Delta(3n^2,m),\quad S_4(m),\quad \Delta(6n^2,m)\quad\mbox{and}\quad \Delta'(6n^2,j,k).
	\end{displaymath}
In the last part of this work we analysed the groups $\mbox{\textlbrackdbl}27,4\mbox{\textrbrackdbl}\cong\zed_9\rtimes\zed_3$ and $S_4(2)\cong A_4\rtimes\zed_4$ in more detail. It turned out that, from the physical point of view, $\mbox{\textlbrackdbl}27,4\mbox{\textrbrackdbl}$ is equivalent to the $SU(3)$-subgroup $\Delta(27)$ and $S_4(2)$ is equivalent to $S_4\times\zed_4$. We closed our discussion with the open question whether this scheme holds true for all finite subgroups of $U(3)$.
\medskip
\\
We hope that this work will shed some light onto the structures of the finite subgroups of $U(3)$, which may be as important in the context of particle physics as the well known finite subgroups of $SU(3)$. 

\subsection*{Acknowledgment}
The author wants to thank Walter Grimus for valuable discussions.

\begin{appendix}

\section{Proofs}

\subsection{Proof of criterion \ref{criterion1}}\label{appendixA}

\begin{proof}
Claim 1: The identity element is the only element mapped onto $\mathbbm{1}_d$ by $D$ $\Leftrightarrow$ $D$ is faithful.
\medskip
\\
$\Leftarrow$: By definition of ``faithful''.
\medskip
\\
$\Rightarrow$: Suppose the identity $e$ is the only element mapped on $\mathbbm{1}_d$ and let $D(a)=D(b)$ for some $a,b\in G$.
	\begin{displaymath}
	D(a)=D(b)\Rightarrow\mathbbm{1}_d=D(a)D(b)^{-1}=D(ab^{-1})
	\Rightarrow ab^{-1}=e\Rightarrow a=b\Rightarrow D\mbox{ injective.}
	\end{displaymath}
$D$ injective $\Rightarrow$ $D$ faithful.
\medskip
\\
Claim 2: Let $M$ be equivalent to a unitary $m\times m$-matrix. Then $\mathrm{Tr}M=m\Leftrightarrow M=\mathbbm{1}_m$.
\medskip
\\
$\Leftarrow$: $\mathrm{Tr}\mathbbm{1}_m=m$.
\medskip
\\
$\Rightarrow$: All eigenvalues $e_j$ of $M$ are elements of $U(1)$, thus
		\begin{displaymath}
                \begin{split}
                  & \left|\mathrm{Tr}M\right|=\left|\sum_{j=1}^m e_j\right|\le \sum_{j=1}^m|e_j|=m.\\
		  & \left|\mathrm{Tr}M\right|=m\Leftrightarrow e_j=\lambda\in U(1)\enspace\forall j \Leftrightarrow M=\lambda \mathbbm{1}_m.\\
		  & \mathrm{Tr}M=\lambda m=m\Rightarrow \lambda=1\Rightarrow M=\mathbbm{1}_m.
                \end{split}
                \end{displaymath}
After all we find: If $D$ is non-faithful there must be more than one element mapped onto $\mathbbm{1}_d$, which is equivalent to the fact that there is more than one character $d$ of $D$ in the character table.
\end{proof}

\subsection{Proof of theorem \ref{theorem1}}\label{appendixB}

\begin{prop}\label{Prop1}
Let $a,b\in\mathbbm{N}\backslash\{0\}$, and let $\mathrm{gcd}(a,b)$ be the greatest common divisor of $a$ and $b$. Then
	\begin{equation}
	\mathcal{Z}_a\cap\mathcal{Z}_b=\{e\}\Leftrightarrow\enspace \mathrm{gcd}(a,b)=1.
	\end{equation}
\end{prop}

\begin{proof} In this proof we represent $\mathcal{Z}_k$ as $\mathcal{Z}_k=\{1,\kappa,\kappa^2,...,\kappa^{k-1}\}$, where $\kappa=e^{\frac{2\pi i}{k}}$.
\begin{itemize}
 \item[$\Rightarrow$:] $\mathcal{Z}_a\cap \mathcal{Z}_b=\{1\}$. Suppose $\mathrm{gcd}(a,b)>1$, and let $\mathrm{lcm}(a,b)$ denote the lowest common multiple of $a$ and $b$.
	\begin{displaymath}
	\Rightarrow \mathrm{lcm}(a,b)=\frac{ab}{\mathrm{gcd}(a,b)}<ab.
	\end{displaymath}
 $\Rightarrow$ $\exists$ $(x,y)\in\{1,...,a-1\}\times\{1,...,b-1\}$ such that
	\begin{displaymath}
	\begin{split}
	\mathrm{lcm}(a,b)=ax=by &\Rightarrow \frac{y}{a}=\frac{x}{b}=\frac{\mathrm{lcm}(a,b)}{ab}=\frac{1}{\mathrm{gcd}(a,b)}<1\Rightarrow\\
		&\Rightarrow e^{2\pi i\frac{y}{a}}=e^{2\pi i\frac{x}{b}}\Rightarrow\\
		&\Rightarrow\underbrace{\left(e^{2\pi i/a}\right)^y=\left(e^{2\pi i/b}\right)^x}_{\in\enspace \mathcal{Z}_a\cap \mathcal{Z}_b}\neq 1.\Rightarrow\mbox{ contradiction!}
	\end{split}
	\end{displaymath}
 \item[$\Leftarrow$:] $\mathrm{gcd}(a,b)=1$. Let $g\in \mathcal{Z}_a\cap \mathcal{Z}_b$.
	\begin{displaymath}
	\begin{split}
	&\Rightarrow \exists (x,y)\in (\mathbbm{N}\backslash\{0\})\times(\mathbbm{N}\backslash\{0\})\enspace\mbox{such that}\\
	&\phantom{\Rightarrow} \enspace g=\left(e^{2\pi i/a}\right)^y=\left(e^{2\pi i/b}\right)^x\Rightarrow\\
	&\Rightarrow \frac{y}{a}=\frac{x}{b}+k,\enspace k\in\mathbbm{N}\Rightarrow yb=xa+kab=(x+kb)a\Rightarrow\\
	&\Rightarrow y\mbox{ is a multiple of }a\mbox{ (because $\mathrm{gcd}(a,b)=1$)}\Rightarrow g=1.
	\end{split}
	\end{displaymath}
\end{itemize}
\end{proof}

\begin{cor}\label{Cor1}
Let $a,b\in\mathbbm{N}\backslash\{0\}$, then
	\begin{equation}
	\zed_a\times\zed_b\cong\zed_{ab}\Leftrightarrow\mathrm{gcd}(a,b)=1.
	\end{equation}
\end{cor}

\begin{proof}
$\alpha:=e^{\frac{2\pi i}{a}}$, $\beta:=e^{\frac{2\pi i}{b}}$, $\gamma:=e^{\frac{2\pi i}{ab}}$.
	\begin{displaymath}
	\zed_{ab}\cong \{1,\gamma,...,\gamma^{ab-1}\}.
	\end{displaymath}
$\zed_a\cong\{1,\alpha,...,\alpha^{a-1}\}$ and $\zed_b\cong \{1,\beta,...,\beta^{b-1}\}$ are normal subgroups of $\zed_{ab}$, thus
	\begin{displaymath}
	\zed_{ab}\cong\zed_a\times\zed_b\Leftrightarrow \zed_a\cap\zed_b=\{1\},
	\end{displaymath}
and from proposition \ref{Prop1}:
	\begin{displaymath}
	\mathcal{Z}_a\cap\mathcal{Z}_b=\{1\}\Leftrightarrow\enspace \mathrm{gcd}(a,b)=1.
	\end{displaymath}

\end{proof}
\hspace{0mm}
\\
\textit{Proof of theorem \ref{theorem1}.} Let us represent $\mathcal{Z}_n$ as $\mathcal{Z}_n=\{1,a,a^2,...,a^{n-1}\}$, where $a=e^{2\pi i/n}$. From proposition \ref{Prop1} we know that
	\begin{displaymath}
	\mathrm{gcd}(n,c)=1\Leftrightarrow \mathcal{Z}_n\cap\mathcal{Z}_c=\{e\}\Leftrightarrow
	\end{displaymath}
	\begin{equation}\label{constructionofZnxG}
	\begin{split}
	\mathcal{D}:\enspace &\mathcal{Z}_n\times G\rightarrow \mathcal{D}(\mathcal{Z}_n\times G)\\
	& (a^k,g)\mapsto a^k D(g),\quad k\in\{0,...,n-1\}
	\end{split}
	\end{equation}
is a faithful representation of $\mathcal{Z}_n\times G$. It remains to show the irreducibility of $\mathcal{D}$. Remember that a representation $R$ of $H$ is irreducible if and only if $(\chi_R,\chi_R)_H=1$ (see equation (\ref{SP-characters}) and explanations there).

	\begin{displaymath}
	\begin{split}
	(\chi_{\mathcal{D}},\chi_{\mathcal{D}})_{\mathcal{Z}_n\times G} &=\frac{1}{\mathrm{ord}(\mathcal{Z}_n\times G)}\sum_{b\in\mathcal{Z}_n\times G}\chi_{\mathcal{D}}(b)^{\ast}\chi_{\mathcal{D}}(b)=\\
	&=\frac{1}{\mathrm{ord}(G)}\times\frac{1}{n}\sum_{k=0}^{n-1}\sum_{b'\in G}\chi_{\mathcal{D}}(a^k b')^{\ast}\chi_{\mathcal{D}}(a^k b')=\\
	&=\frac{1}{\mathrm{ord}(G)}\times\frac{1}{n}\times n\sum_{b'\in G}\chi_{D}(b')^{\ast}\chi_{D}(b')=(\chi_D,\chi_D)_G=1.
	\end{split}
	\end{displaymath}
$\Rightarrow$ $\mathcal{D}$ is a faithful $m$-dimensional irreducible representation of $\mathcal{Z}_n \times G$.

\begin{flushright}
$\Box$\end{flushright}

\subsection{Proofs of theorem \ref{GroupSeriesTheorem} and theorem \ref{GroupCenterTheorem}}\label{appendixC}

\begin{lemma}\label{lemma1}
Let $A:=\lang A_1,...,A_a\rang$ be a normal subgroup of $G:=\lang A_1,...,A_a,B\rang$, then
	\begin{equation}
	\mathrm{ord}(G)\leq\mathrm{ord}(A)\mathrm{ord}(B).
	\end{equation}
\end{lemma}

\begin{proof}
$A$ is an invariant subgroup of $G$, thus every element of $G$ can be written as an element of $A$ times an element of $\lang B\rang$, i.e. a power of $B$. Thus there are at most $\mathrm{ord}(A)\mathrm{ord}(B)$ different elements in $G$.
\end{proof}
\hspace{0mm}\\
\textit{Proof of theorem \ref{GroupCenterTheorem}}.
Every element of $\lang A_1,...,A_a,B\rang$ can be written as a product of generators of the group. Let
	\begin{equation}\label{defP}
	P(A_1,...,A_a,B)=\one_m.
	\end{equation}
be a representation of the unit element in terms of generators. The number $x[P]$ of factors $B$ contained in the product $P$ can be defined by
	\begin{equation}
	P(A_1,...,A_a,\delta B)=\delta^{x[P]} P(A_1,...,A_a,B)
	\end{equation}
for some $\delta\in U(1),\enspace \delta^n\neq 1\enspace\forall n\in\mathbbm{N}\backslash\{0\}$. Let $M$ be the set of all products $P$ fulfilling equation (\ref{defP}). We can now define $s$ to be the smallest positive number of factors $B$ contained in a product (\ref{defP}), i.e.
	\begin{equation}
	s:=\min_{P\in M} \{x[P]|x>0\}.
	\end{equation}
Since $B^n=\one_m$ we find $s\leq n$. Let $\tilde{P}$ denote a product of generators fulfilling
	\begin{equation}
	\tilde{P}(A_1,...,A_a,\delta B)=\delta^s \tilde{P}(A_1,...,A_a,B).
	\end{equation}
Suppose $s<n$: The center of the group is generated by $e^{2\pi i/c}\one_m$, $(\beta B)^{n^k}=e^{2\pi i/c^j}\one_m$ and $\beta^s\one_m=\tilde{P}(A_1,...,A_a,\beta\one_m)$.
	\begin{displaymath}
	(\beta^{s})^{c^j}=\left(e^{2\pi i/n^k}\right)^s,\mbox{ which (if $s<n$) generates }\zed_{n^k}. 
	\end{displaymath}
Thus we find
	\begin{displaymath}
	\lang A_1,...,A_a,\beta B\rang=\lang A_1,...,A_a,B,e^{2\pi i/c^j}\one_m,e^{2\pi i/n^k}\one_m\rang
	\end{displaymath}
and since $\lang e^{2\pi i/c}\one_m\rang\subset\lang A_1,...,A_a\rang$:
	\begin{equation}\label{equ1}
	\mathrm{ord}(\lang A_1,...,A_a,\beta B\rang)=\mathrm{ord}(\lang A_1,...,A_a,B\rang)\times c^{j-1}n^k=c^{j-1}n^{k+1}\mathrm{ord}(H).
	\end{equation}
On the other hand we know that
	\begin{displaymath}
	\lang A_1,...,A_a,\beta B\rang=\lang A_1,...,A_a,e^{2\pi i/c^j}\one_m,XB\rang
	\end{displaymath}
	for some $X\in\lang e^{2\pi i/n^k}\one_m\rang$. Using lemma \ref{lemma1} this leads to
	\begin{displaymath}
	\mathrm{ord}(\lang A_1,...,A_a,\beta B\rang)\leq \underbrace{\mathrm{ord}(\lang A_1,...,A_a,e^{2\pi i/c^j}\one_m\rang)}_{c^{j-1}\mathrm{ord}(H)}\times\mathrm{ord}(\lang XB\rang)\leq c^{j-1}n^k\mathrm{ord}(H),
	\end{displaymath}
which is a contradiction to equation (\ref{equ1}) $\Rightarrow s=n$.
\medskip
\\
Since $s=n$ the center of the group is given by
	\begin{displaymath}
	\lang e^{2\pi i/c^j}\one_m, e^{2\pi i/n^{k-1}}\one_m\rang\cong\zed_{c^j n^{k-1}}\quad\quad\mbox{for $j>0,\enspace k>0$}.
	\end{displaymath}
The other cases follow immediately noticing that (by definition) $\lang e^{2\pi i/c}\one_m\rang$ always is a subgroup of the center.
\medskip
\\
Let us finally consider the case $c=1,\enspace k>0$. In this case the group
	\begin{displaymath}
	\lang A_1,...,A_a,\beta B\rang,\quad\quad \beta=e^{2\pi i/n^k}
	\end{displaymath}
(by construction) is the semidirect product
	\begin{displaymath}
	\lang A_1,...,A_a\rang\rtimes\lang \beta B\rang\cong H\rtimes\zed_{n^k}
	\end{displaymath}
with center $\lang e^{2\pi i/n^{k-1}}\one_m\rang\cong \zed_{n^{k-1}}$.
\begin{flushright}
$\Box$\end{flushright}

\begin{lemma}\label{lemma2}
Let $n,q\in\mathbbm{N}\backslash\{0\}$ and $\mathrm{gcd}(n,q)=1$. Then
	\begin{equation}
	\exists p\in \{1,...,n-1\}: (pq)\hspace{1mm}\mathrm{mod}\hspace{1mm}n=1.
	\end{equation}
\end{lemma}

\begin{proof}
Consider the numbers
	\begin{displaymath}
	q\hspace{1mm}\mathrm{mod}\hspace{1mm}n,\enspace (2q)\hspace{1mm}\mathrm{mod}\hspace{1mm}n,\enspace...\enspace,\enspace [(n-1)q]\hspace{1mm}\mathrm{mod}\hspace{1mm}n.
	\end{displaymath}
Suppose 
	\begin{displaymath}
	(k_1 q)\mathrm{mod}\hspace{1mm}n=(k_2 q)\mathrm{mod}\hspace{1mm}n,\quad\quad k_1,k_2\in\{1,...,n-1\},\quad k_1>k_2.
	\end{displaymath}
This implies
	\begin{displaymath}
	[\underbrace{(k_1-k_2)}_{\in \{1,...,n-2\}} q]\hspace{1mm}\mathrm{mod}\hspace{1mm}n=0\Rightarrow q\hspace{1mm}\mathrm{mod}\hspace{1mm}n=0\Rightarrow\mbox{ contradiction to }\mathrm{gcd}(n,q)=1.
	\end{displaymath}
$\Rightarrow$ The $n-1$ numbers
	\begin{displaymath}
	q\hspace{1mm}\mathrm{mod}\hspace{1mm}n,\enspace (2q)\hspace{1mm}\mathrm{mod}\hspace{1mm}n,\enspace...\enspace,\enspace [(n-1)q]\hspace{1mm}\mathrm{mod}\hspace{1mm}n
	\end{displaymath}
are different elements of $\{1,...,n-1\}$. $\Rightarrow$ One of them must be $1$.
\end{proof}
\hspace{0mm}
\\
\textit{Proof of theorem \ref{GroupSeriesTheorem}.}
Let $b=qc^j n^k$; $\enspace j,k\in\mathbbm{N}$; $\enspace\mathrm{gcd}(q,n)=\mathrm{gcd}(q,c)=1$.
\medskip
\\
From lemma \ref{lemma2} we know that
	\begin{displaymath}
	\exists p\in\{1,...,n-1\}: (pq)\hspace{1mm}\mathrm{mod}\hspace{1mm}n=1.
	\end{displaymath}
	Then
	\begin{displaymath}
	(\beta B)^{pq}=\beta^{pq}B^{pq}=\beta^{pq}B
	\end{displaymath}
	and we can write $\beta B$ as a product of the two group elements $(\beta^{pq}B)^{-1}\beta B=\beta^{1-pq}\one_m$ and $\beta^{pq}B$. Therefore
	\begin{displaymath}
	\lang A_1,...,A_a,\beta B\rang= \lang A_1,...,A_a,\beta^{pq}B,\beta^{1-pq}\one_m\rang.
	\end{displaymath}
	$(\beta^{1-pq})^{c^j n^k}=e^{2\pi i/q}$, thus $\beta^{1-pq}\one_m$ generates $\zed_{rq}\cong \zed_{r}\times\zed_q$, where $r$ contains factors $n$ and $c$ only.
	\begin{displaymath}
	\begin{split}
	\Rightarrow\lang A_1,...,A_a,\beta B\rang & =\lang A_1,...,A_a, \beta^{pq}B,e^{2\pi i/r}\one_m,e^{2\pi i/q}\one_m\rang\cong\\
	&\cong\lang A_1,...,A_a,\beta^{pq} B,e^{2\pi i/r}\one_m\rang\times \lang e^{2\pi i/q}\one_m\rang\cong\\
	&\cong\lang A_1,...,A_a,\beta^{pq} B,e^{2\pi i/r}\one_m\rang\times\zed_q.
	\end{split}
	\end{displaymath}
	$\Rightarrow$ If we want that $\lang A_1,...,A_a,\beta B\rang$ can not be written as a direct product with a cyclic group we must impose $q=1$, thus
	\begin{displaymath}
	b=c^j n^k,\quad j,k\in\mathbbm{N}.
	\end{displaymath}
It remains to show that for $b=c^j n^k$ $\lang A_1,...,A_a,\beta B\rang$ can not be written as a direct product with a cyclic group.
\medskip
\\
Let from now on $\beta:=e^{2\pi i/(c^j n^k)}$. Suppose
	\begin{displaymath}
	\lang A_1,...,A_a,\beta B\rang=X\times Y,
	\end{displaymath}
where $Y$ is a cyclic group. Because of irreducibility $Y$ must be a subgroup of the center $C$ of the group. In the following we will frequently use the fact that $X\cap Y=\{\one_m\}$ in $X\times Y$.
\medskip
\\
Let us first consider the case $j>0,\enspace k>1$. From theorem \ref{GroupCenterTheorem} we know that
	\begin{displaymath}
	C=\lang e^{2\pi i/c^j}\one_m,e^{2\pi i/n^{k-1}}\one_m\rang\cong \zed_{c^j}\times\zed_{n^{k-1}}.
	\end{displaymath}
Since every element of $X\times Y$ can be uniquely written as a product of an element of $Y$ and an element of $X$ it follows that
	\begin{displaymath}
	\exists\enspace\alpha\one_m\in Y:\quad \alpha\beta B\in X.
	\end{displaymath}
	\begin{displaymath}
	\Rightarrow (\alpha\beta B)^{n^{k-1}}=\alpha^{n^{k-1}}e^{2\pi i/(c^j n)}\one_m\in X.
	\end{displaymath}
Since $Y\subset C\cong \zed_{c^j}\times\zed_{n^{k-1}}$ and $\alpha\one_m\in Y$ we find that $\alpha^{n^{k-1}}\one_m\in\lang e^{2\pi i/c^j}\one_m\rang$, which implies
	\begin{displaymath}
	(\alpha\beta B)^{c^j n^{k-1}}=(\alpha^{n^{k-1}}e^{2\pi i/(c^j n)})^{c^j}\one_m=e^{2\pi i/n}\one_m\in X,
	\end{displaymath}
thus $Y\cap \lang e^{2\pi i/n}\one_m\rang=\{\one_m\}\Rightarrow Y$ is a subgroup of $\lang e^{2\pi i/c^j}\one_m\rang$. In the cases of $j>0,\enspace k\in\{0,1\}$ we find $Y\subset \lang e^{2\pi i/c^j}\one_m\rang$ too. The case of $c=1$ directly leads to $Y=\{\one_m\}$.
\medskip
\\
Knowing that $Y\subset \lang e^{2\pi i/c^j}\one_m\rang$ we can deduce that (if $Y$ is nontrivial)
	\begin{displaymath}
	\lang e^{2\pi i/c}\one_m\rang\subset Y\Rightarrow \lang e^{2\pi i/c^j}\one_m\rang\cap X=\{\one_m\},
	\end{displaymath}
because else we would find $X\cap Y\neq\{\one_m\}$. This leads to
	\begin{displaymath}
	Y=\lang e^{2\pi i/c^j}\one_m\rang.
	\end{displaymath}
Thus every element of $\lang A_1,...,A_a,\beta B\rang$ can be uniquely written as a product of an element of $X$ and an element of $Y=\lang e^{2\pi i/c^j}\one_m\rang$. This implies that every element of $\lang A_1,...,A_a\rang$ can be uniquely written as an element of some subgroup $S\subset X$ and an element of $Y=\lang e^{2\pi i/c^j}\one_m\rang$.
	\begin{displaymath}
	\begin{split}
	& \Rightarrow \lang A_1,...,A_a\rang=S\times\lang e^{2\pi i/c}\one_m\rang\Rightarrow\\ & \Rightarrow G=\lang A_1,...,A_a,B\rang\cong(S\times\lang e^{2\pi i/c}\one_m\rang)\rtimes\zed_n\cong(S\rtimes\zed_n)\times\zed_c,
	\end{split}
	\end{displaymath}
which is a contradiction to ``$G\cong H\rtimes\zed_n$ can not be written as a direct product with a cyclic group''.
\medskip
\\
The case $j=0$, in a similar way, leads to $Y=\lang e^{2\pi i/c}\one_m\rang$ leading to the same contradiction as above.

\begin{flushright}
$\Box$\end{flushright}

\end{appendix}

\end{document}